\pdfoutput=1
\documentclass[11pt,letterpaper]{article}
\usepackage{jheppub}
\usepackage{amsthm}
\def\ket#1{\langle #1\rangle}
\usepackage{multirow}
\usepackage{array}

\newtheorem*{mainconjecture}{Main Conjecture}
\newtheorem*{weakbasisconjecture}{Weak Basis Conjecture}
\newtheorem*{strongbasisconjecture}{Strong Basis Conjecture}
\newtheorem{lemma}{Lemma}
\newtheorem{corollary}{Corollary}

\title{The Soft-Collinear Bootstrap:\\
$\mathcal{N}=4$ Yang-Mills Amplitudes
at Six- and Seven-Loops}

\author[a]{J.~L.~Bourjaily,}
\author[b]{A.~DiRe,}
\author[c]{A.~Shaikh,}
\author[b]{M.~Spradlin}
\author[b]{and A.~Volovich}

\affiliation[a]{Department of Physics,
Harvard University, Cambridge, MA 02138}

\affiliation[b]{Department of Physics,
Brown University, Providence, RI 02912}

\affiliation[c]{School of Engineering,
Brown University, Providence, RI 02912}

\abstract{
Infrared divergences in scattering amplitudes arise when a loop momentum $\ell$ becomes collinear with a massless external momentum $p$.
In gauge theories, it is known
that the $L$-loop logarithm of a planar amplitude
has much softer infrared singularities than the $L$-loop amplitude itself.
We argue that planar amplitudes in ${\cal N} = 4$
super-Yang-Mills theory enjoy softer than expected behavior
as $\ell \parallel p$ already at the level of the integrand.
Moreover, we conjecture that the four-point
integrand can be uniquely determined,
to any loop-order, by imposing the correct soft-behavior of the logarithm together with dual conformal invariance and dihedral symmetry.
We use these simple criteria to determine explicit formulae for the
four-point integrand through seven-loops, finding perfect agreement
with previously known results through five-loops.
As an input to this calculation, we enumerate all
four-point dual conformally invariant (DCI) integrands
through seven-loops, an analysis which
is aided by several graph-theoretic
theorems we prove about general DCI integrands at arbitrary loop-order.
The six- and seven-loop amplitudes receive non-zero contributions from
229 and 1873 individual DCI diagrams respectively.
\\
\\
PDF and {\sc Mathematica} files with all of our results are provided at
\href{http://goo.gl/qIKe8}{{\tt http://goo.gl/qIKe8}}
}

\begin{document}

\maketitle

\section{Introduction}

It has long been appreciated that
direct Feynman diagram calculations are often a very inefficient way to study
multi-loop scattering amplitudes, especially in a theory as
simple as planar ${\cal N} = 4$ super-Yang-Mills (SYM) theory.
Rather, one usually aims first to obtain a representation
of a desired amplitude
as a linear combination of a (hopefully) small number of (hopefully)
relatively simple integrals.
To date there have been at least
four essentially different technologies available
for determining an integral representation for an amplitude.
In this paper we introduce a new approach which is both conceptually simple
and computationally powerful, as we demonstrate by using it to determine
the seven-loop four-particle integrand in SYM theory.

\subsection{Brief Review of Four Paths Towards Integral Representations}

The most conventional approach is to begin with a collection of integrands, make
the ansatz that the amplitude should be a linear combination of those
integrals, and then determine the coefficients of the integrands
by matching various data---such as generalized unitarity
cuts~\cite{Bern:1994cg,Bern:1994zx,Bern:2007ct} or
leading
singularities~\cite{Buchbinder:2005wp,Cachazo:2008vp,Cachazo:2008hp,Spradlin:2008uu,Spradlin:1900zz}---between the ansatz and the amplitude.
Each `data point' generates a linear equation on the coefficients, so for
sufficiently many data points, and
a sufficiently large
set of linearly independent
integrands, one can
obtain a unique solution for the coefficients.
Actually establishing the correctness of the resulting ansatz (i.e., ruling
out the possibility of additional contributions to the amplitude which
happen to vanish on all of the data points considered) usually
requires careful analysis, for example via more complicated
$D$-dimensional unitarity cuts.
Details of this tried and true
approach can be found in the review~\cite{Bern:2011qt}, and we have
learned that its power has been exploited in as yet unpublished
work to construct an integral representation for the four-particle
amplitude at six-loops~\cite{unpublished_six_loop}.

A very different approach has been
introduced in~\cite{CaronHuot:2010zt,Boels:2010nw}, in particular
in~\cite{ArkaniHamed:2010kv}, where a relation was presented
which allows one (in principle) to write down the integrand of any
desired multi-loop amplitude in SYM theory recursively
in terms of lower-loop integrands (involving successively more particles).
Note that we distinguish here between {\it the} integrand, which is {\it a uniquely-defined rational function of internal and external kinematic data in any
planar theory}, and {\it an} integral representation, which is
understood to be well-defined only modulo
the addition of terms which integrate to zero.

A third even more recent approach relies on the duality, or equivalence,
between certain correlation functions in SYM theory and maximally helicity violating (MHV)
scattering amplitudes
at the level of the integrand~\cite{Eden:2010zz}.  Aspects of this duality
have been studied in several
papers including~\cite{Alday:2010zy,Eden:2010ce,Eden:2011yp,Eden:2011ku}.
For the special case of four particles the relevant correlation functions
possess an extra hidden symmetry which has been exploited
to provide a slick derivation of the four-particle integrand at
three-loops
in~\cite{Eden:2011we}, and at higher loops in~\cite{Eden:2012tu}.

A fourth approach makes use of the understanding of how soft- and collinear-infrared singularities factorize and exponentiate in planar gauge
theory~\cite{Akhoury:1978vq,Mueller:1979ih,Collins:1980ih,Sen:1982bt,Sterman:1986aj,Catani:1989ne,Collins:1989bt,Magnea:1990zb,Catani:1990rp,Giele:1991vf,Kunszt:1994np,Catani:1998bh,Vogt:2000ci,Sterman:2002qn}.
If we restrict for simplicity our attention to
MHV amplitudes and denote by
$M^{(L)}_n = A^{(L)}_n/A^{(0)}_n$ the ratio of the $L$-loop $n$-particle
MHV (color-stripped, partial)
amplitude to the corresponding tree amplitude, then exponentiation
implies that in dimensional regularization with $D=4 - 2\epsilon$ we have
\begin{equation}
\label{eq:log}
\log\left[1 + \lambda M^{(1)}_n + \lambda^2 M^{(2)}_n + \cdots \right]
= {\cal O}(\epsilon^{-2}),
\end{equation}
even though the individual contributions on the left-hand side
have stronger infrared (IR)
singularities, $M^{(L)}_n = {\cal O}(\epsilon^{-2L})$.
The general form of this statement is regulator-independent: for example,
in the Higgs regulator of~\cite{Alday:2009zm}, the right-hand side
of eqn.~(\ref{eq:log}) would be ${\cal O}(\log^2 m)$ while each
individual
$M^{(L)}_n$ diverges as ${\cal O}(\log^{2L} m)$
as the mass is taken to zero---a fact which has
been used to guide the construction of various
integrands in the Higgs-regularized theory~\cite{Henn:2010bk, Henn:2010ir}.

The divergent behavior of equation~(\ref{eq:log})
has often been used as an important consistency check on the correctness
of integral
representations~\cite{Bern:2006vw,Cachazo:2006tj,Bern:2006ew,Bern:2008ap,Cachazo:2008hp,Spradlin:2008uu}.
Moreover it can be used
to guide the construction of a representation
for an $L$-loop amplitude
by using the requirement that all of the
${\cal O}(\epsilon^{-2L})$ through ${\cal O}(\epsilon^{-3})$ poles
must cancel as `data' in the same sense as described above.
For example, for the four-particle amplitude at
three-loops, if one starts with the one- and two-loop
amplitudes as given, and makes the ansatz that $M_4^{(3)}$
should be a linear combination of the three-loop ladder and tennis court
diagrams,
then their coefficients can be uniquely
fixed by requiring that eqn.~(\ref{eq:log}) should hold at order $\lambda^3$.
In this case, the resulting ansatz
agrees with the correct three-loop
amplitude found in~\cite{Bern:2005iz}, though in general
this approach is not guaranteed to work,
since it cannot detect the contribution
from any individual integral whose infrared behavior is too soft.

\subsection{Preview of a New Path Towards Integrands}

The last approach outlined above has the obvious disadvantage that
it relies upon explicit evaluation of multi-loop integrals, a problem
which---despite recent progress---remains extremely difficult in general.
In this paper we explore an approach which amounts essentially to
imposing eqn.~(\ref{eq:log}) at the level of the {\it integrand}, which allows
us to work with simple rational functions instead of complicated
polylogarithm functions.
(Indeed it has been noted in~\cite{Drummond:2010mb} that soft limits
can also be handled very easily at the level of the
integrand,
in particular an $(n+1)$-particle integrand should reduce directly
to an $n$-particle integrand in the soft limit $p_n \to 0$.
This statement is not generally true for integrals due to IR-divergences.)

Infrared divergences arise when some loop momentum becomes collinear
with an external momentum $p_a$.  In terms of the dual coordinates $x_a$
defined by $p_a = x_{a+1} - x_{a}$ this happens whenever some loop integration
variable $x_A$ approaches the
line connecting $x_a$ and $x_{a+1}$ for some $a$
(see~\cite{ArkaniHamed:2010kv,ArkaniHamed:2010gh} for a detailed
discussion).
If we take a limit where both $(x_a - x_A)^2$ and $(x_{a+1} - x_A)^2$
are order $\epsilon$ (but making sure that all other kinematic
invariants take generic, non-vanishing values)
then any integrand will have
a $1/\epsilon^2$ pole in the limit $\epsilon \to 0$.
However, we claim that the integrand of the {\it logarithm} of any
amplitude has only a $1/\epsilon$ pole at any loop-order $L > 1$.
This phenomenon is clearly related to
the fact of
eqn.~(\ref{eq:log}),
but at the moment we present this claim only as an empirical observation
which we have verified for the four-point integrand
through seven loop-order.

Assuming that this behavior
holds for all $L$, then it obviously can be used as a
powerful tool for constructing integrands of amplitudes.
We simply need to postulate a suitable basis of integrands for
any desired amplitude and try to find a linear combination for
which the ${\cal O}(1/\epsilon^2)$ poles cancel.
We call this procedure the soft-collinear bootstrap because it allows the
determination of an $L$-loop integrand given all of the corresponding
integrands at loop-order $<L$, which appear in the $L$-loop logarithm.
(This approach has a similar flavor to the work of ref.~\cite{Bargheer:2009qu},
where it was argued that all tree-level amplitudes could be fully
constrained by their collinear singularities, combined with
the requirement of Yangian symmetry).

The case of four particles is special because a very simple basis naturally
suggests itself: the collection of integrands which are invariant under
conformal transformations on the dual $x$ variables, appropriately
called dual conformally invariant
(DCI) integrands~\cite{Drummond:2007cf,Drummond:2007au}.
This class of integrands has already played an important role
in guiding the construction
of integral representations at four-~\cite{Bern:2006ew}
and five-loops~\cite{Bern:2007ct}.
We have checked through seven-loops that the family of
DCI integrands at each loop-order is linearly independent (if one imposes
dihedral symmetry in the external particles, which is a symmetry of
all superamplitudes; it is curious to note that imposing mere cyclic symmetry is sufficient to fix the amplitude through five-loops, but the full dihedral symmetry must be imposed in order to obtain the six- and seven-loop amplitudes).
Moreover
they remain linearly independent even after picking-off from each one
just the
residue of its $1/\epsilon^2$ pole.
Taken together, we are therefore led to
conjecture that the four-point integrand can be uniquely
determined, to any loop-order, by imposing the ${\cal O}(1/\epsilon)$
collinear behavior together with dual conformal invariance and dihedral
symmetry.

It is of course well-known that the four-point amplitude in SYM theory
is trivial, in the sense that its logarithm
is determined exactly by the anomalous
dual conformal Ward identity~\cite{Drummond:2007au}
up to an overall multiplicative factor of the cusp anomalous
dimension $f(\lambda)$~\cite{Korchemsky:1985xj}, whose value is in turn known exactly from
other considerations~\cite{Beisert:2006ez}.
However, since $f(\lambda)$ can obviously be computed by simply
integrating the integrand,
it is amusing to note
that while (anomalous) dual conformal invariance
of the
logarithm of the amplitude is not powerful enough to fix $f(\lambda)$ alone,
our results suggest that
dual conformal invariance of the integrand is, when one also
imposes the mild ${\cal O}(1/\epsilon)$ collinear behavior together with dihedral symmetry with respect to the external particles.
Moreover despite the triviality of the amplitude,
the integrand of its logarithm is of independent interest since it
is evidently
a nontrivial function which makes other remarkable appearances,
for example as a Wilson loop expectation value in twistor
space~\cite{Mason:2010yk,CaronHuot:2010ek,Belitsky:2011zm,Adamo:2011pv}
and as (the square root of) a certain correlation function
of stress tensor
multiplets~\cite{Alday:2010zy,Eden:2010zz,Eden:2010ce,Eden:2011yp,Eden:2011ku}.

In section~\ref{sec:two} we review important facts and conventions
related to integrands in SYM theory and present our main conjectures (all of
which we have verified through seven-loops).  In section~\ref{sec:classification}
we explain the classification of DCI diagrams, which we use as a basis for
constructing the four-particle integrands.  Our results are
presented and discussed in
section~\ref{sec:results}, and the reader may find more information
at~\cite{four_point_multiloop_datafiles}.

Throughout this paper we use the momentum twistor
parameterization~\cite{Hodges:2009hk}.
Consequently
all momenta (both internal and external) are strictly
four-dimensional and so all Gram determinant conditions are automatically
satisfied.
However we suspect that when expressed in terms of the $x$'s, formulae for
four-point integrands
are actually valid
for planar SYM theory in any number of dimensions, a property which is
known to hold at least through four-loops~\cite{Bern:2010tq}.

All of the important ingredients in this paper,
including the exponentiation of
infrared divergences and dual conformal invariance, rely crucially
on planarity, but it would clearly be of great interest if some variant
of our approach might be applicable beyond the planar limit.
Integral representations for the complete four-point amplitude in
SYM theory, including
all non-planar contributions, are known through four-loops
(see for example~\cite{Bern:2012uf}).

\section{Conventions and Conjectures}
\label{sec:two}

\subsection{The Integrand}
\label{sec:theintegrand}

The integrand
of the four-point amplitude $M_4^{(L)} = A_4^{(L)}/A_4^{(0)}$
in planar SYM theory is,
at $L$-loop-order, a rational function involving only the
scalar quantities
\begin{equation}
x_{ab} = (x_a - x_b)^2,\qquad\mathrm{with}\qquad s\equiv x_{24}\quad\mathrm{and}\quad t\equiv x_{13}
\end{equation}
formed from the
dual coordinates
$x_1,x_2,x_3,x_4$ associated with the external kinematics and the $L$
loop integration variables $x_A, x_B, \ldots$.
For example at one and two-loops we have~\cite{Green:1982sw,Anastasiou:2003kj}
\begin{align}
\label{eq:oneloop}
M_4^{(1)}(x_A)
&= \frac{x_{13} x_{24}}{x_{1A} x_{2A} x_{3A} x_{4A}}\,,\\
\label{eq:twoloop}
M_4^{(2)}(x_A,x_B) &=
\frac{x_{13}^2 x_{24}}{x_{1A} x_{1B} x_{2B} x_{3A} x_{3B} x_{4A} x_{AB}} +
\frac{x_{13} x_{24}^2}{x_{1B} x_{2A} x_{2B} x_{3A} x_{4A} x_{4B} x_{AB}} +
(A \leftrightarrow B)\,.
\end{align}
In order to avoid a little bit of clutter we omit throughout
this paper the somewhat conventional overall prefactors of
$(-1/2)^L$ and $1/L!$---the latter from symmetrization with respect to the loop integration variables, and the former coming
from a choice of how to normalize the measure of integration.
In the interest of specificity, let us note that to compute amplitudes in the normalization
convention of for
example~\cite{Anastasiou:2003kj} our $L$-loop integrands should be integrated with the measure
\begin{equation}
\left( - \frac{1}{2}\right)^L \frac{1}{L!} \int \prod_{k=1}^L \frac{d^4x_{A_k}}{i \pi^2}\,.
\end{equation}
The above examples exhibit three important general properties of
integrands:
\begin{enumerate}
\item{full permutation symmetry $\mathfrak{s}_L$ in the $L$
integration variables $x_A,x_B,\ldots$;}
\item{dihedral symmetry $D_4$ in the external variables $x_1,\ldots,x_4$;}
\item{and the absence of double poles.}
\end{enumerate}

\newpage
\subsection{The Integrand of the Logarithm}

The above properties ensure that there is no ambiguity in defining
the integrand of the logarithm of an amplitude. Taylor-expanding the left-hand side of eqn.~(\ref{eq:log}) to seventh-order in $\lambda$, we have\begin{equation}
\label{eq:logtwo}
\begin{split}
(\log M_n)^{(2)} =&\phantom{\,+\,}M_n^{(2)} - \frac{1}{2} (M_n^{(1)})^2\,;\\
(\log M_n)^{(3)} =&\phantom{\,+\,}M_n^{(3)} - M_n^{(2)} M_n^{(1)} + \frac{1}{3} (M_n^{(1)})^3\,;\\
(\log M_n)^{(4)} =&\phantom{\,+\,}M_n^{(4)} - M_n^{(3)} M_n^{(1)} - \frac{1}{2} (M_n^{(2)})^2 + M_n^{(2)}(M_n^{(1)})^2 -\frac{1}{4}(M_n^{(1)})^4\,;\\
(\log M_n)^{(5)} =&\phantom{\,+\,}M_n^{(5)} - M_n^{(4)} M_n^{(1)}  -M_n^{(3)} M_n^{(2)} + M_n^{(3)}(M_n^{(1)})^2+(M_n^{(2)})^2M_n^{(1)}\\
& -M_n^{(2)}(M_n^{(1)})^3+\frac{1}{5} (M_n^{(1)})^5\,;\\
(\log M_n)^{(6)} =&\phantom{\,+\,} M_n^{(6)} - M_n^{(5)} M_n^{(1)}  -M_n^{(4)} M_n^{(2)} + M_n^{(4)}(M_n^{(1)})^2-\frac{1}{2}(M_n^{(3)})^2
\\
&+2M_n^{(3)}M_n^{(2)}M_n^{(1)}-M_n^{(3)}(M_n^{(1)})^3+\frac{1}{3}(M_n^{(2)})^3-\frac{3}{2}(M_n^{(2)})^2(M_n^{(1)})^2\\
&+M_n^{(2)}(M_n^{(1)})^4-\frac{1}{6} (M_n^{(1)})^6\,;\\
(\log M_n)^{(7)} =&\phantom{\,+\,}M_n^{(7)} - M_n^{(6)} M_n^{(1)}  -M_n^{(5)} M_n^{(2)} + M_n^{(5)}(M_n^{(1)})^2-M_n^{(4)}M_n^{(3)}\\
&+2M_n^{(4)}M_n^{(2)}M_n^{(1)}-M_n^{(4)}(M_n^{(1)})^3+(M_n^{(3)})^2M_n^{(1)}+M_n^{(3)}(M_n^{(2)})^2\\
&-3M_n^{(3)}M_n^{(2)}(M_n^{(1)})^2+M_n^{(3)}(M_n^{(1)})^4\\&-(M_n^{(2)})^3M_n^{(1)}+2(M_n^{(2)})^2(M_n^{(1)})^3-M_n^{(2)}(M_n^{(1)})^5+\frac{1}{7} (M_n^{(1)})^7\,.
\end{split}
\end{equation}
For the two- and three-loop logarithms, these expressions are to be interpreted at the level of the integrand respectively as
\begin{equation}
\begin{split}
(\log M_n)^{(2)}(x_A,x_B)&= M_n^{(2)}(x_{\left[A\right.},x_{\left. B\right]}) -
\frac{1}{2} M_n^{(1)}(x_{\left[A\right.}) M_n^{(1)}(x_{\left. B\right]})\\
&= M_n^{(2)}(x_A,x_B) -
\frac{1}{2}\Big(M_n^{(1)}(x_A) M_n^{(1)}(x_B)+M_n^{(1)}(x_B) M_n^{(1)}(x_A)\Big),
\end{split}
\end{equation}
and
\begin{equation}\label{eq:logthree}
\begin{split}
(\log M_n)^{(3)}(x_A,x_B,x_C)=&\phantom{\,+\,}M_n^{(3)}(x_{\left[A\right.},x_B,x_{\left. C\right]})-M_n^{(2)}(x_{\left[\left[A\right.\right.},x_{\left.B\right]})M_n^{(1)}(x_{\left. C\right]})\\
&+ \frac{1}{3} M_n^{(1)}(x_{\left[A\right.}) M_n^{(1)}(x_B) M_n^{(1)}(x_{\left.C\right]}) \\
=&\phantom{\,+\,}M_n^{(3)}(x_A,x_B,x_C)-\Big(M_n^{(2)}(x_A,x_B)M_n^{(1)}(x_C)\\
&+M_n^{(2)}(x_A,x_C)M_n^{(1)}(x_B)+M_n^{(2)}(x_B,x_C)M_n^{(1)}(x_A)\Big)\\
&-\frac{1}{3}\Big(M_n^{(1)}(x_A)M_n^{(1)}(x_B)M_n^{(1)}(x_C)+5\mathrm{\,other\,permutations}\Big),
\end{split}
\end{equation}
where we have used the notation $\left[A\cdots C\right]$ to denote the outer-symmetrization over indices $A,\ldots,C$ (being careful to avoid over-counting when symmetrizing the product of separately symmetrized functions).
Each of these expressions manifests the properties outlined at the end of section~\ref{sec:theintegrand}.

\subsection{The Main Conjecture}

In order to phrase our main conjecture precisely, it is best to employ
the momentum twistor variables of Hodges~\cite{Hodges:2009hk}.
These parameterize the on-shell external momenta in terms of
four-component momentum twistors $Z_a^i$, $a=1,\ldots,n$, $i=1,\ldots,4$,
which are related to
the scalar momentum invariants according to
\begin{equation}
\label{eq:mt}
x_{ab} = (x_b - x_a)^2=\frac{\ket{a\,a+1\,\,b\,b+1}}{\ket{a\,a+1}\ket{b\,b+1}} \propto \ket{a\,a{+}1\,\,b\,b{+}1}\,.
\end{equation}
Here (as usual) all indices $a,b,\ldots$ are understood modulo
$n$ and the bracket
denotes the determinant
\begin{equation}
\ket{a\,b\,c\,d} \equiv \mathrm{det}\left(Z_a Z_b Z_c Z_d\right)\,.
\end{equation}
Finally the `$\propto$' in~(\ref{eq:mt}) indicates that we will not keep track of the two-brackets in the denominator.  This is well-justified because all dependence on them drops out of any dual conformally invariant function of the $x_{ab}$.

Each off-shell loop integration variable $x_A$ is parameterized by
the antisymmetric combination $x_{A_1} \wedge x_{A_2}$
of a pair of four-component
momentum twistors. These can appear in brackets of the form
\begin{equation}
\label{eq:mttwo}
x_{aA} = (x_a- x_A)^2 \to \ket{a\,a{+}1\,A_1\,A_2}, \quad
x_{AB} = (x_A - x_B)^2 \to \ket{A_1\,A_2\,B_1\,B_2}, \quad {\rm or~etc.,}
\end{equation}
where, as forewarned above, we omit the two-brackets completely since our intention is to
only ever
use these replacements in dual conformally invariant formulae.

We can probe the infrared structure of an amplitude at the level of the
integrand by taking the limit as one of the loop integration variables
approaches the line connecting $x_a$ and $x_{a+1}$ for some $a$,
as discussed in~\cite{ArkaniHamed:2010kv,ArkaniHamed:2010gh}.
Due to the symmetry of the integrand
it is sufficient to consider the limit as $x_A$ approaches the line connecting
$x_1$ and $x_2$.
This can be accomplished by taking $Z_{A_1}$ to $Z_2$ and $Z_{A_2}$ to
any generic point which lies on the hyperplane spanned by $Z_1,Z_2,Z_3$, i.e.
\begin{equation}
Z_{A_1} \to Z_2+\mathcal{O}(\epsilon), \qquad Z_{A_2} \to \alpha Z_1 + \beta Z_2 + \gamma Z_3\,+\mathcal{O}(\epsilon).
\end{equation}
In this limit, we clearly have
\begin{equation}
x_{1A} = \ket{1\,2\,A_1\,A_2} \propto \epsilon\,, \qquad
x_{2A} = \ket{2\,3\,A_1\,A_2} \propto \epsilon\,.
\end{equation}
By `generic' we mean that we don't want to accidentally choose
$Z_{A_2}$ so that some other singularities are probed at the same time;
in particular
this means that we must take both $\alpha$ and $\gamma$ to be nonzero,
for if $\alpha$ were zero then
$x_{3A} = \langle 3\,4\,A_1\,A_2 \rangle$
would also vanish, while if $\gamma$ were zero then
$x_{4A} = \langle 4\,1\,A_1\,A_2 \rangle$
would also vanish.

It might be interesting to attempt to
glean further information about integrands
by probing these multi-collinear regions, but at the moment we content
ourselves with simple collinear limits. With this in mind, let us consider the generally safe choice of taking $\alpha,\beta,\gamma$ all to be $1$:
\begin{equation}\label{eq:mainlimit}
Z_{A_1} \to Z_2+\mathcal{O}(\epsilon), \qquad Z_{A_2} \to  Z_1 +  Z_2 +  Z_3\,+\mathcal{O}(\epsilon)\,.
\end{equation}

In this region of the $A_1$ integral, the statement that an integrand has at most an $\mathcal{O}(1/\epsilon)$-divergence can be formalized as the conjecture:
\begin{mainconjecture}
The integrand of the logarithm of the $n$-particle $L$-loop amplitude
in planar SYM theory behaves as ${\cal O}(1/\epsilon)$ in the
limit~(\ref{eq:mainlimit}) for all $L>1$.
\end{mainconjecture}

\subsection{Two- and Three-Loop Examples}\label{two_and_three_loop_examples}

The integrand of the 2-loop logarithm $(\log M_4)^{(2)}$ is the sum of
5 terms, obtained by plugging eqns.~(\ref{eq:oneloop}) and~(\ref{eq:twoloop})
into eqn.~(\ref{eq:logtwo}).
Only the three terms containing
both $x_{1A}$ and $x_{2A}$ in the denominator contribute to the ${\cal O}(1/\epsilon^2)$ pole
in the limit~(\ref{eq:mainlimit}).  Combining these three terms over a common denominator yields
\begin{equation}
(\log M_4)^{(2)}(x_A,x_B) =
\frac{x_{13} x_{24} \left( x_{1B} x_{24} x_{3A} + x_{13} x_{2B} x_{4A} - x_{13} x_{24} x_{AB}
\right)}{x_{1A} x_{1B} x_{2A} x_{2B} x_{3A} x_{3B} x_{4A} x_{4B} x_{AB}}
+ {\cal O}(1/\epsilon)\,.
\end{equation}
Turning to momentum twistors, we can plug in
\begin{align}
x_{24}&= \ket{2\,3\,4\,1} = - \ket{1\,2\,3\,4}\,,\\
x_{3A}&= \ket{3\,4\,A_1\,A_2} = - \ket{1\,2\,3\,4} + {\cal O}(\epsilon)\,,\\
x_{13}&= \ket{1\,2\,3\,4}\,,\\
x_{4A}&= - \ket{1\,2\,3\,4} + {\cal O}(\epsilon)\,,
\end{align}
so the nontrivial numerator factor becomes
\begin{equation}
x_{1B} x_{24} x_{3A} + x_{13} x_{2B} x_{4A} - x_{13} x_{24} x_{AB}
= \ket{1\,2\,3\,4}^2 \left(
x_{1B} - x_{2B} - x_{AB}
\right) + {\cal O}(\epsilon)\,,
\end{equation}
but this in turn is easily seen to vanish as a consequence of
\begin{equation}
x_{AB} = \ket{A_1\,A_2\,B_1\,B_2} = \ket{2\,1\,B_1\,B_2} + \ket{2\,3\,B_1\,B_2}
= - x_{1B} + x_{2B}\,.
\end{equation}

The same cancellation of the ${\cal O}(1/\epsilon^2)$ pole occurs at
three-loops, but let us turn the statement around by pretending for
a moment that we did not know the integrand, but were willing
to make the ansatz that it should be a linear combination,
\begin{equation}
M^{(3)}_4(x_A,x_B,x_C) = a\, I^{(3)}_1 + b\, I^{(3)}_2\, ,
\end{equation}
of the {\it only} two
available dual conformally invariant integrands
\begin{align}
\label{eq:threeladder}
I^{(3)}_1 &=
\frac{x_{13}^3 x_{24}}{x_{1A} x_{1B} x_{1C} x_{2C} x_{3A} x_{3B} x_{3C} x_{4B} x_{AB} x_{AC}} + {\rm sym.~(12~terms~total)}\,, \\
I^{(3)}_2 &=
\frac{x_{13}^2 x_{24} x_{2A}}{x_{1A} x_{1C} x_{2B} x_{2C} x_{3A} x_{3B} x_{4A} x_{AB} x_{AC} x_{BC}} +  {\rm sym.~(24~terms~total)}\,.
\label{eq:tennis}
\end{align}
Assembling eqn.~(\ref{eq:logthree})
leads to an expression for $(\log M_4)^{(3)}(x_A,x_B,x_C)$ as a sum
of 36 terms.
If we isolate those terms containing both $x_{1A}$ and $x_{2A}$ we
find that the integrand vanishes at order ${\cal O}(1/\epsilon^2)$ in
the limit~(\ref{eq:mainlimit}) only for $a = b = 1$, in accord
with the known value for the three-loop integrand~\cite{Bern:2005iz}.

\subsection{Integral Basis Conjectures for $n=4$}

As exemplified in the previous paragraph, we can use our conjecture
as a tool
for determining an integrand, but only if we first identify a
basis of integrands for constructing a suitable ansatz.
For $n>4$ there are two separate potential problems.  First of all, because MHV amplitudes for $n>4$ are {\it maximally} chiral, they cannot be represented by the parity-even quantities $x_{ab}$ alone; rather they {\it must} involve  parity-sensitive generalizations of $x_{ab}$ such as,
\begin{equation}
\ket{A_1\,A_2\,2\,4}
\end{equation}
which could perhaps be written ``$x_{\star\,A}$'' where $x_{\star}$ were defined as {\it one} (of the two) complex points in space-time simultaneously light-like separated from all four momenta $p_1,\ldots, p_4$. Secondly, the class of multi-loop integrands involving such factors tend to satisfy many {\it integrand-level} relations (that is, they generally are not linearly independent but rather form an over-complete basis of integrands).

It is well-known that
scalar box integrals provide a basis for {\it integral} representations
of all one-loop amplitudes
in SYM theory, modulo terms which integrate to zero in four
dimensions, and it has recently been shown that certain chiral octagons
provide a
basis (though overcomplete)
for all one-loop integrands~\cite{ArkaniHamed:2010gh}.
The construction of
a basis for planar two-loop integral representations has been
carried out in~\cite{Gluza:2010ws}, and it would be interesting
to explore the implications of that work for the special case of the integrand
of planar SYM theory.

However for the special case of $n=4$ we expect a very simple story:
namely,
that the integrand can be expressed
as a linear combination of rational functions of the
$x_{ab}$, each of which is invariant under dual conformal
transformations~\cite{Drummond:2006rz} (in particular, under the simultaneous inversion
$x^\mu \to x^\mu/x^2$ of all external and internal dual variables).
Indeed this property was used to guide the construction of the four-
and five-loop amplitudes in~\cite{Bern:2006ew} and~\cite{Bern:2007ct},
where such integrals were called `pseudo-conformal'.

Dual conformal invariance is broken by the dimensional regularization
traditionally used to regulate the infrared divergences of loop-amplitudes,
so one might expect that imposing dual conformal invariance at the level of
the integrand is not exactly synonymous to imposing it at the level of
the integrated amplitude.
Indeed, it was observed in~\cite{Drummond:2007aua} that only a particular
subclass of pseudo-conformal integrals actually appear, namely those
which can be rendered infrared-finite by evaluating them in a particular
scheme (off-shell regularization) which manifestly preserves dual
conformal invariance.
We use the abbreviation DCI to refer to an integrand which is both
pseudo-conformal in the
sense of~\cite{Bern:2007ct} and finite when regulated off-shell as
in~\cite{Drummond:2007aua}.

The observation of~\cite{Drummond:2007aua} through five-loops together
with our analysis through
seven-loops motivates us to make the following conjecture:

\begin{weakbasisconjecture}
At any loop-order, the integrand of the four-particle amplitude in planar
SYM theory can be expressed as a unique linear combination of dual conformally
invariant integrands involving only the quantities $(x_a-x_b)^2$.
\end{weakbasisconjecture}

\noindent
Let us emphasize again that this conjecture
makes two logically separate claims:
first of all, that for any $L$, the integrand can be expressed
in terms of this particular class of objects;
and secondly, that this collection
is linearly independent for any $L$.
Also, we note that this conjecture is stronger than a corresponding
statement about dual conformally invariant {\it integrals}.  For example,
if two quantities differ by something which integrates to zero, they
would be linearly independent as integrands but not as integrals.

In order for the method to have maximal power, which is to say
in order for it to be able to uniquely determine the coefficient of
every single DCI integrand, it is necessary for us to make the slightly stronger conjecture:

\begin{strongbasisconjecture}
At any loop-order, the collection of DCI integrands
remains linearly independent
if we isolate from each one the residue of the $1/\epsilon^2$
pole in the limit~(\ref{eq:mainlimit}).
\end{strongbasisconjecture}

\noindent
We have verified that this is true through seven-loops by explicit calculation.

\section{A Basis for the Integrand}
\label{sec:classification}

In this section we outline the classification of the dual conformally invariant
diagrams we use as a
(conjectured) basis for constructing four-particle integrands.

\subsection{Pseudo-Conformal Diagrams}
\label{sec:pseudo}

Since all quantities constructed from the $x_{ab}$'s are automatically
invariant under translations and rotations in $x$-space, the only nontrivial
constraint from dual conformal invariance arises from imposing invariance
under inversions, $x^\mu_a \to x^\mu_a/x^2$, which takes
\begin{equation}
\label{eq:inversion}
x_{ab} = (x_a - x_b)^2 \to \frac{(x_a - x_b)^2}{x_a^2 x_b^2} =
\frac{x_{ab}}{x_a^2 x_b^2}\,, \qquad
d^4 x \to \frac{d^4 x}{(x^2)^4}\,.
\end{equation}
In order to check whether a given rational function of the $x_{ab}$'s is
invariant under this transformation, one simply has to check that all of the
$x_a^2$ factors which accumulate as a consequence of eqn.~(\ref{eq:inversion})
cancel out.  This is tantamount to simply counting how many times a given
index appears in the denominator and numerator.  For example, for the
first term of the
two-loop integrand~(\ref{eq:twoloop}),
\begin{equation}
\frac{x_{13}^2 x_{24}}{x_{1A} x_{1B} x_{2B} x_{3A} x_{3B} x_{4A} x_{AB}}\,,
\end{equation}
we note that indices 1 and 3 each appear twice both upstairs and downstairs,
while indices 2 and 4 each appear once both upstairs and downstairs, so all
of their associated $x_a^2$ weights cancel out.  At the same time the
indices $A$ and $B$ associated to the integration variables each appear
four times downstairs, which is expected as their weights should cancel
against the factors arising from transforming the measure $d^4 x_A\,d^4 x_B$.

\begin{figure}[t]\caption{Dual conformally invariant diagrams through three-loops.
We use the standard notation where a dotted line connecting two faces
$x_a,x_b$ corresponds to a numerator factor $x_{ab}$, but for clarity
we omit from each diagram an overall factor of $x_{13} x_{24}$.
These four diagrams correspond to the rational functions
given in eqns.~(\ref{eq:oneloop}), (\ref{eq:twoloop}), (\ref{eq:threeladder})
and~(\ref{eq:tennis}) respectively.}
\begin{center}\includegraphics[scale=1]{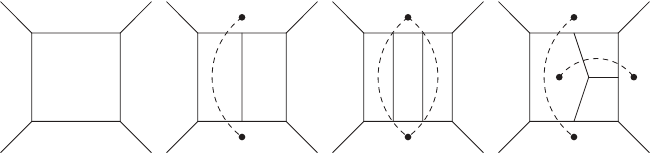}\vspace{-0.4cm}\end{center}
\label{fig:firstfour}
\end{figure}

At this level of analysis we are looking simply at the na\"{i}ve transformation
of integrands under~(\ref{eq:inversion}), with no consideration given
to the question of infrared divergences which arise when the integral is
actually performed.
Following~\cite{Bern:2007ct} we therefore refer to any quantity invariant
under~(\ref{eq:inversion}) as pseudo-conformal.
In figure~\ref{fig:firstfour} we review
a standard diagrammatic notation for expressing the class of
pseudo-conformal quantities of interest.
(However at high loop-order we abandon this notation because the
proliferation of numerator factors renders most diagrams illegible.)
Two minor differences with respect to the standard
notation are that: (1) for clarity we omit from each diagram an overall factor
of $x_{13} x_{24}$ (we prove in Lemma~\ref{lemma:one} of appendix~A that every
four-point DCI diagram has
such an overall factor); and (2) for each diagram we impose (as emphasized
above)
by construction
that the resulting
rational function of $x_{ab}$'s should have $D_4 \times \mathfrak{s}_L$ symmetry and
should be normalized so that the numerical factor in front of each
{\it distinct} term in the sum is precisely $1$.

The translation between a diagram and its associated rational function of
$x_{ab}$'s proceeds via the following simple steps.
We label the four external faces
$x_1,x_2,x_3,x_4$ in cyclic order and then label the $L$ internal faces
of an $L$-loop diagram as $x_A,x_B,\ldots$.
For example, let us choose to label the fourth diagram in
figure~\ref{fig:firstfour} as:
\begin{center}
\includegraphics[scale=0.666667]{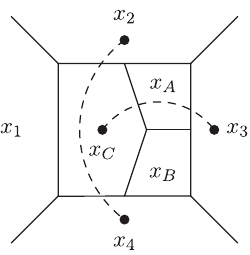}
\end{center}
Each propagator bounded by
faces $x_a$ and $x_b$ is associated
as usual to a factor $1/x_{ab}$, while each dotted
line connecting faces $x_a$ and $x_b$ is associated to
a numerator factor $x_{ab}$.
(Note that, as in the third diagram in figure~\ref{fig:firstfour}, there may be $k>1$
dotted lines connecting the same two faces, in which case the numerator
factor $x_{ab}$ is raised to the $k$-th power.)
This diagram has 10 propagators and 2 numerator factors, so according
to the rules it gives the factor
\begin{equation}
\label{eq:aaa}
x_{13} x_{24} \frac{x_{24} x_{3C}}{x_{1C} x_{2A} x_{2C} x_{3A} x_{3B} x_{4B}
x_{4C} x_{AB} x_{AC} x_{BC}}\,,
\end{equation}
where we have restored the overall factor of $x_{13} x_{24}$ mentioned above.

The final step is to impose $D_4 \times \mathfrak{s}_3$ symmetry
by summing the quantity~(\ref{eq:aaa}) over all dihedral transformations
of the external $x_1,x_2,x_3,x_4$ and all permutations of the internal
$x_A,x_B,x_C$.
Applying this procedure to~(\ref{eq:aaa}), and keeping in mind our choice
to normalize each distinct term with a factor of $1$,
leads immediately to the expression given in~(\ref{eq:tennis}).
It is important to note that because we always impose the $D_4 \times \mathfrak{s}_L$
symmetry at the level of the integrand, the choice of explicit loop momenta labeling `$x_A, x_B,\ldots$' in these figures is largely irrelevant: they only serve as a representative with which we may present a particular numerator explicitly; any other choice would have been acceptable, as it would also generate the same rational function after symmetrization. (The labels of loop momenta with only four propagators are particularly irrelevant, as they do not appear in the numerator at all, and so are never needed to specify an explicit, representative integrand).

\subsection{Dual Conformally Invariant Diagrams}
\label{sec:divergent}

It is known that there are $1,1,2,10,59$ distinct pseudo-conformal
diagrams respectively at loop-orders $L=1,2,3,4,5$, respectively.  However, only
$1,1,2,8,34$ of those diagrams actually enter the $L$-loop four-particle
integrand with nonzero coefficient; the remaining 2 four-loop and 25 five-loop
diagrams do not appear in the amplitude.
In~\cite{Drummond:2007aua} it was observed that a precise characterization
can be given which distinguishes the contributing versus non-contributing
diagrams: the former consist of those pseudo-conformal diagrams which
are rendered infrared finite when they are evaluated off-shell, i.e.~when
we take the external momenta $p_i$ to satisfy $p_i^2 \ne 0$.
(Note that there are additional types of DCI diagrams one can draw
when off-shell legs are allowed.  These have been classified
through four-loops in~\cite{Nguyen:2007ya}, but these additional
diagrams play no role in our present analysis.)

To illustrate this point let us consider the simplest five-loop
diagram which is pseudo-conformal but not DCI,
\begin{center}
\includegraphics[scale=0.666667]{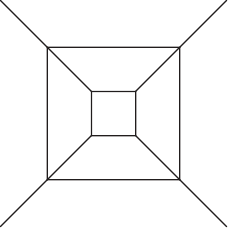}
\end{center}
whose associated integrand is
\begin{equation}
\frac{x_{13} x_{24}}{
x_{1E} x_{2C} x_{3D} x_{4B} x_{AB} x_{AC} x_{AD} x_{AE} x_{BD} x_{BE}
x_{CD} x_{CE}} + {\rm sym.~(120~terms~total)}\,.
\end{equation}
Let us focus on the region of integration where all four of $x_B,x_C,x_D$ and
$x_E$ simultaneously approach $x_A$.  Specifically let us suppose
that $x_{AB}, x_{AC}, x_{AD}, x_{AE}$ are all order ${\cal O}(\rho^2)$ as $\rho \to 0$.
In this limit there are eight vanishing propagators, giving a factor of
$\rho^{16}$ in the denominator, and the integral behaves as
\begin{equation}
\sim
\int \frac{x_{13} x_{24}}{x_{1A} x_{2A} x_{3A} x_{4A}} \frac{d^4 x_A\ 
\rho^{15}\,d\rho}{\rho^{16}} + \cdots\,,
\end{equation}
where the factor $\rho^{15}\,d\rho$ in the numerator comes from switching
to radial coordinates for $x_B,x_C,x_D,x_E$ near the point where all of
these equal $x_A$.  Evidently this integral has a divergence near
$\rho=0$.

In contrast none of the diagrams shown in figure~\ref{fig:firstfour}, nor
any of the 8 (34) diagrams which are known to contribute to the 4 (5)-loop
integrands, suffer from any divergences of this sort.
We refer to pseudo-conformal diagrams which possess this property as {\it genuinely
dual conformally invariant} (DCI) diagrams.
Let us emphasize that this stands as an empirical observation, as there is
no rigorous understanding of why this peculiar fact
should play any important
role.  In particular let us emphasize that DCI diagrams most certainly do
have infrared singularities when evaluated in any of the familiar
regularization schemes such as dimensional regularization
or the Higgs regulator of~\cite{Alday:2009zm}.
Nevertheless, appreciating the utility of this finiteness criterion through
five-loops, we adopt the conjecture that it continues to hold to all-loop
order.

\subsection{Classification of Four-Point DCI Diagrams Through Seven-Loops}

In this section we explain the steps taken to classify
all DCI diagrams through seven-loops.
The classification takes place in two steps:  we first used an
off-the-shelf program to generate all relevant distinct planar topologies,
and then wrote our own code to determine, for each topology,
all independent sets of
numerator
factors which render the diagram dual conformally invariant.
The vast majority of planar graphs (in particular, any which contain
a triangle---
an internal face with only three edges) cannot be rendered DCI
with any
numerator factor, and can be rejected immediately.  On the other hand,
beginning at five-loops there are examples of planar graphs
which are rendered dual conformally invariant
by more than one inequivalent choice of numerator
factor (the three occurrences of this phenomenon at five-loops are evident in fig.~\ref{five_loop_integrand}).
Counting diagrams with different numerator factors separately
we find $1,1,2,8,34,256,2329$ distinct four-point DCI diagrams
at one- through seven-loops respectively.  These come from
$1,1,2,8,30,197,1489$ distinct underlying topologies (ignoring numerators).
These results are summarized in table~\ref{integrand_statistics}.

\begin{table}[t]\vspace{-0cm}\caption{Statistics of four-point DCI integrands through seven-loops.\label{integrand_statistics}}\vspace{0.1cm}
\begin{tabular}{|l|r|r|r|r|r|r|}
\cline{1-7}\multirow{2}{*}{{\bf Loop}}&\multicolumn{1}{c|}{{\bf \# of denominator}}&\multicolumn{1}{c|}{{\bf \# of distinct}}&\multicolumn{4}{c|}{{\bf \# of integrands with coefficient:}}\\\cline{4-7}
&\multicolumn{1}{c|}{{\bf topologies}}&\multicolumn{1}{c|}{{\bf DCI integrands}}&$\qquad\mathbf{+1}$&$\qquad\mathbf{-1}$&$\qquad\mathbf{+2}$&$\qquad\mathbf{0}$\\\hline
1&1&1&1&0&0&0\\\hline
2 & 1 & 1 & 1 & 0 & 0 & 0 \\\hline
3 & 2 & 2 & 2 & 0 & 0 & 0 \\\hline
4 & 8 & 8 & 6 & 2 & 0 & 0 \\\hline
5 & 30 & 34 & 23 & 11 & 0 & 0 \\\hline
6 & 197 & 256 & 129 & 99 & 1 & 27\\\hline
7 & 1489 & 2329 & 962 & 904 & 7 & 456\\\hline
\end{tabular}\vspace{-0.2cm}
\end{table}

The first step in the classification requires the choice of a graph generating program.
The popular
choice {\tt qgraf}~\cite{Nogueira:1991ex,qgraf}, is not well-suited
to the present application because of our interest exclusively in plane
graphs.
Note that at
this point we must begin to be precise about distinguishing
a `{\it planar} graph', which is one that {\it admits} an embedding onto the plane
with no self-intersection, from a `{\it plane} graph', which is a planar graph
together with a {\it specific choice} of plane-embedding.  This distinction is important
because a given planar graph can have more than one inequivalent embedding.
Two plane graphs which correspond to the same
underlying planar graph would be considered
identical in scalar field theory but must be distinguished in
gauge theory because they have different color index flow.

This rather subtle distinction first becomes necessary at six-loops,
where we find the candidate planar graph shown in fig.~\ref{fig:crazy} which admits
two inequivalent plane-embeddings, only one of which can be made pseudo-conformal
with a suitable numerator factor.
Note that in this case the pseudo-conformal graph on the left
is not DCI (it fails the finiteness criterion explained in
section~\ref{sec:divergent}).  We strongly suspect that it is not
possible for any genuine DCI diagram to admit an inequivalent
embedding in the plane, but this notion plays no role in our analysis and
we only mention this phenomenon in passing as a curiosity.

\begin{figure}[b]\begin{center}
\includegraphics[scale=0.45]{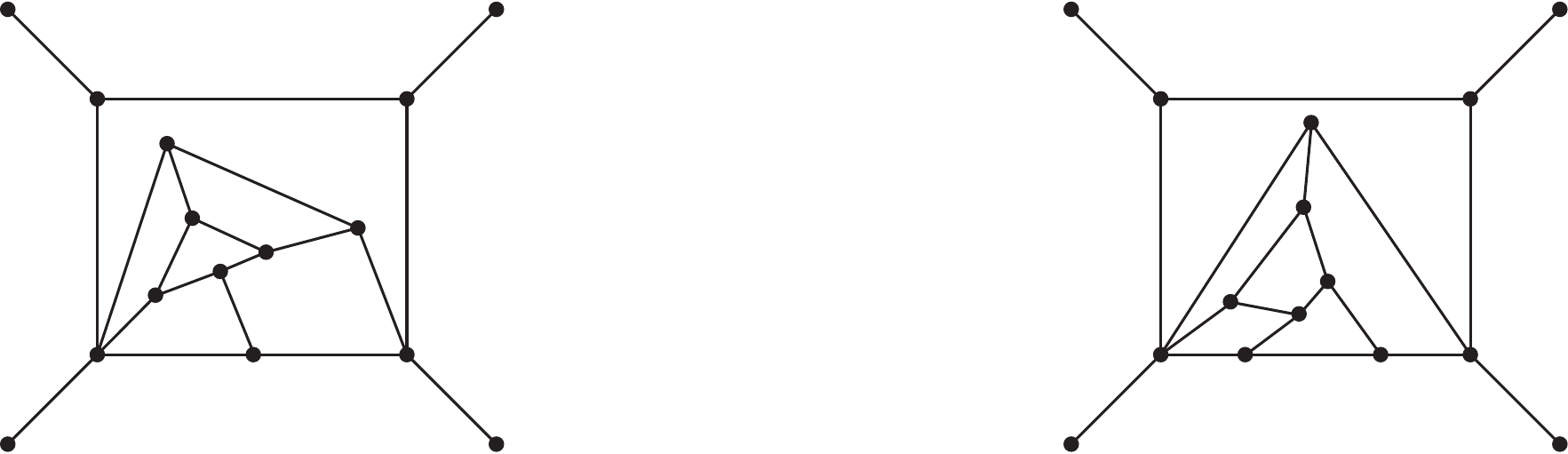}\vspace{-0.4cm}
\end{center}\caption{Two different plane embeddings of the same planar graph (they are related to each other by
rotating the bottom edge of the square by 180 degrees out of the plane of the paper while leaving
the other three edges fixed).
The graph on the left can be made pseudo-conformal with a suitable numerator factor while the graph
on the right cannot.}

\label{fig:crazy}
\end{figure}

Returning to our main point:
in order to both avoid wasting effort generating non-planar graphs in
which we are not interested, and to be sure of generating
all possible inequivalent plane graphs, we turn to the specialized
program {\tt plantri}~\cite{plantripaper,plantri} which is capable of exactly
what we require.  The program has no concept of external lines, so
to represent a four-point diagram in {\tt plantri} we first `fold it up'
to a plane graph with $(L+4)$ faces by
tying the four external lines together at
a new `vertex at infinity.'
A given $(L+4)$-face plane graph, together with the choice of a
quadrivalent vertex in the graph, may then be `unfolded' to
a planar four-point diagram.  The process of folding and unfolding is
not bijective since on the one hand two inequivalent four-point graphs may become
identical when folded up, while on the other hand a given plane graph with $k>1$
quadrivalent vertices may unfold to fewer than $k$ topologically distinct four-point
diagrams.
The latter happens when a plane graph has a discrete
symmetry which exchanges two (or more) of its quadrivalent vertices.

To generate all topologies of interest we:
\begin{enumerate}
\item{used {\tt plantri} to enumerate all
plane graphs $G_1,G_2,\ldots$ with vertex connectivity
$\ge 1$ and at least one quadrivalent vertex;}
\item{unfolded each graph $G_i$ at each of its quadrivalent vertices to obtain
various four-point diagrams $G_i^{(1)}, G_i^{(2)},\ldots$ (taking
care to remove any duplicates, as mentioned above);}
\item{from the collection of diagrams so generated, we excluded any
which have internal triangles, since these cannot possibly be made DCI.}
\end{enumerate}
It turns out that the {\tt plantri} output is most naturally sorted not by $L$ but by the number of
vertices in the graph.  An $L$-loop DCI graph can have up to
$2L+3$ vertices (including the vertex at infinity),
so in order to classify all diagrams through seven-loops we need all graphs
with up to 17 vertices.
The reader may be interested to know that we found
exactly
\begin{equation}
1,0,1,1,6,13,71,308,1637,8421,45229,243220,1326433
\end{equation}
`DCI candidate' plane graphs with 5 through 17 vertices respectively.
By a DCI candidate we mean a plane graph which has at least one
quadrivalent vertex with the property that unfolding the graph at that
point leads to a triangle-free four-point diagram for the denominator.

The four-point diagrams obtained from these DCI candidates were then
fed into a custom Java program which scanned through all possible
inequivalent
numerator factors, selecting those which make the diagram both pseudo-conformal
according to section~\ref{sec:pseudo} and finite according
to the criterion reviewed in section~\ref{sec:divergent}.
Many of the DCI candidates do not admit any valid numerator,
while a few admit more than one inequivalent choice of numerator.
The final output from our code program was a list of the $1,1,2,8,34,256,2329$
distinct DCI integrands through seven-loops advertised above; refer also to table~\ref{integrand_statistics}.

(These numbers include 1 6-loop and 2 7-loop four-point diagrams
with vertex connectivity 1, which have vertex connectivity 2 when
folded up.  These diagrams do not contribute to scattering amplitudes.
The 6-loop example was discussed in
figure~12 of~\cite{Bern:2007ct}.)

As part of this work we in fact generated all DCI
candidate plane graphs through 18
vertices, but found it computationally prohibitive to tackle 19,
which would have been necessary for a complete classification at eight-loops.
While an $L$-loop DCI graph can have as many as $2L+3$ vertices,
it can on the other hand have as few as $L + 4$ (examples saturating
this lower bound exist for $L > 3$).
Therefore, although our classification of DCI diagrams is complete
only through seven-loops, the data we have amassed contains
an enormous
number of additional DCI diagrams through fourteen-loops, which we have
not yet analyzed.

\section{Bootstrapping the Four-Point Integrand Through Seven-Loops}
\label{sec:results}

The way in which the four-point multi-loop integrand can be obtained via the soft-collinear bootstrap was illustrated for the three-loop integrand in section \ref{two_and_three_loop_examples}. The condition that $(\log M_4)^{(L)}$ has at most an $\mathcal{O}(1/\epsilon)$-pole in the soft-collinear limit (\ref{eq:mainlimit}) is equivalent to the condition that the {\it numerator} of $(\mathrm{log}~M_4)^{(L)}$ (when all contributing terms are assembled over a common denominator) must vanish in this limit (as the denominator is certainly proportional to $\epsilon^2$).  With this in mind, the procedure for obtaining the $L$-loop integrand $M_4^{(L)}$ from the soft-collinear bootstrap can be carried out as follows:
\begin{enumerate}
\item expand the log of the $L$-loop amplitude---as illustrated in eqn.~(\ref{eq:logtwo})---and compute the contribution of each term involving strictly lower-loop amplitudes to the overall numerator in the soft-collinear limit (\ref{eq:mainlimit}); call this `$\mathrm{log}_{\mathrm{rest}}$';
\item find {\it all} $L$-loop DCI integrands $I_{i}^{(L)}$ (with each one fully symmetrized with respect to permutations of the loop momentum variables $A_1,\ldots, A_L$, and the $D_4$ dihedral symmetry of the external momenta), and compute the contribution of each to the overall numerator in the soft-collinear limit (\ref{eq:mainlimit}); call these contributions $(I_i^{(L)})_{\mathrm{soft}}$;
\item $\mathrm{log}_{\mathrm{rest}}$ and the $(I_i^{(L)})_{\rm soft}$ are of course still nontrivial polynomials in various $x_{ab}$'s, and our goal is now to determine a collection of {\it numerical constants} $c_i$ such that \[\mathrm{log}_{\mathrm{rest}} + \sum_i c_i(I_i^{(L)})_{\mathrm{soft}}=0.\]  By repeatedly evaluating this equation at sufficiently many random independent values of the remaining variables ($x_{A_2},\ldots,x_{A_L}$ as well as $x_1,\ldots,x_4)$ this equation can be turned into a linear system on the $c_i$ with constant coefficients.  The existence and uniqueness of the solution, which we find through seven-loops, are the central miracles of this paper.
\end{enumerate}

Our conjecture is that this general strategy extends to all loop-orders.  This procedure is simple enough to automate, and can be implemented efficiently enough to find all integrands through six-loops in less than 2 minutes on a sufficiently powerful computer (given knowledge of the $L$-loop DCI integrands $I_i^{(L)}$).

For the sake of providing a uniform reference for the reader, the results for four and five-loops are reviewed in figures~\ref{four_loop_integrand}, and~\ref{five_loop_integrand}, respectively.
For six- and seven-loops, there are too many contributions to be practically
included graphically here, but they may be viewed at~\cite{four_point_multiloop_datafiles}.

Beyond five-loops, there are some interesting features worth noting, which are summarized in table \ref{integrand_statistics}. For example, although all 34 five-loop DCI integrands contribute to the amplitude, this is not generally the case: at six-loops, 27 of the 256 DCI integrands have vanishing coefficients; and at seven-loops, 456 of the 2329 DCI have vanishing coefficients. The 27 six-loop DCI integrands with vanishing coefficient are illustrated in figure~\ref{six_loop_coeff_0}.

More curiously, at six-loops there is the first appearance of a coefficient found to be `$+2$'; and at seven-loops, seven integrands have coefficient $+2$. These are illustrated in figures~\ref{six_loop_coeff_2} and \ref{seven_loop_coeff_2}, respectively.
It would be fascinating to see if this pattern of coefficients could be understood
along the lines of the method proposed in~\cite{Cachazo:2008dx}.
\begin{figure}[b]
\begin{center}\includegraphics[scale=0.45]{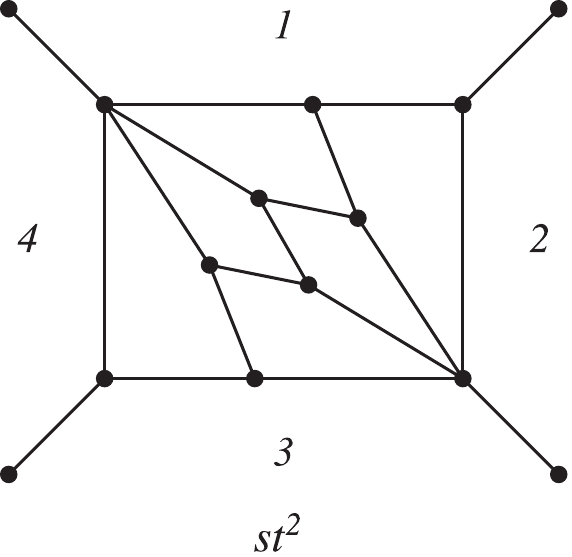}
\caption{The only six-loop integrand contributing with coefficient $+2$.\label{six_loop_coeff_2}}
\end{center}
\end{figure}

\begin{figure}[t]
\caption{The 7 seven-loop contributions with coefficient $+2$.\label{seven_loop_coeff_2}}~\\
\noindent\mbox{\hspace{-0.05\textwidth}
\includegraphics[width=1.1\textwidth]{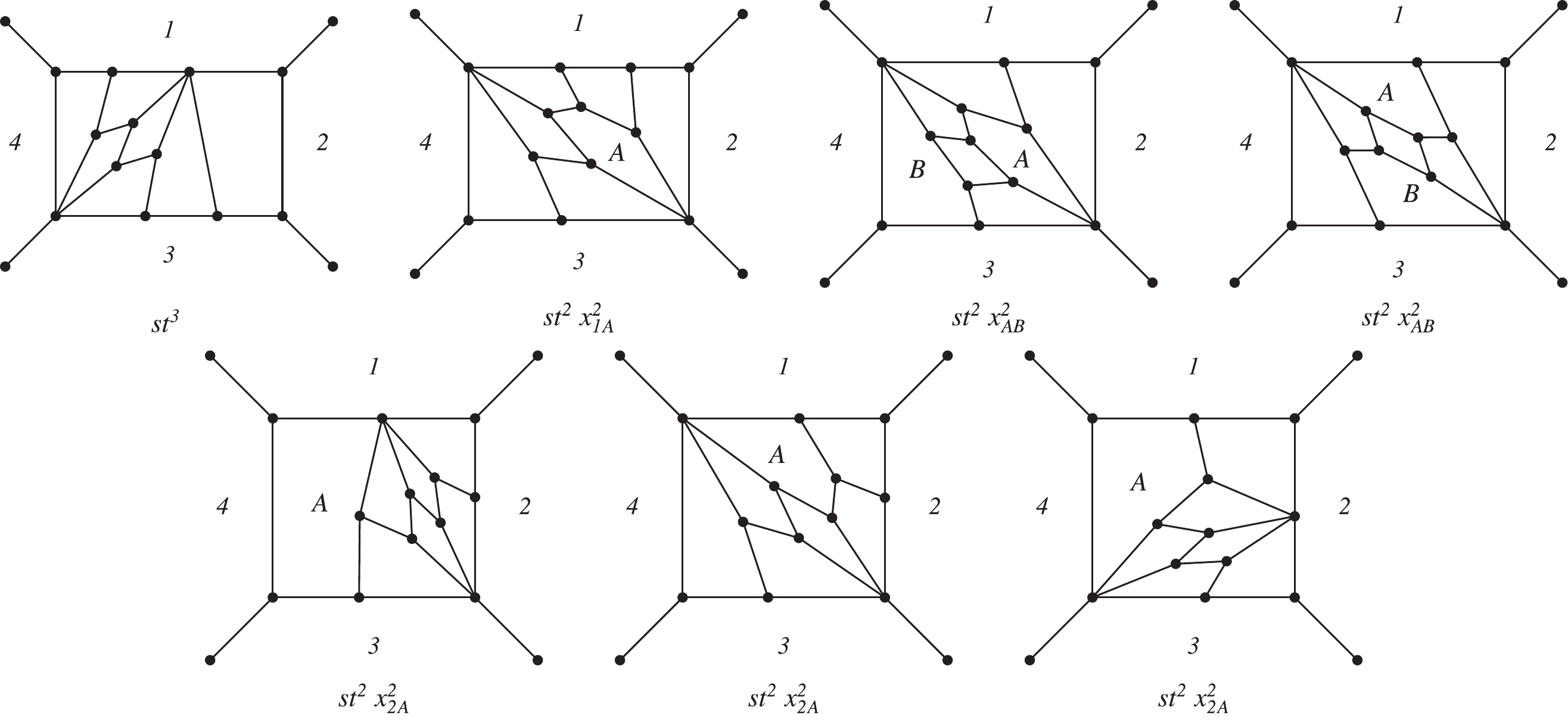}}
\end{figure}

One available and relatively simple partial check on our results involves two-particle cuts, which
have been shown to iterate to all orders in SYM theory~\cite{Bern:1997nh}.  This means that if there is
an $L$-loop contribution to the integrand with the property that it can be reduced to a product
of an $L-1$-loop contribution and the four-particle tree amplitude by cutting only two internal propagators,
then that $L$-loop diagram must inherit the same coefficient as its $L-1$-loop ancestor.
At two- and three-loops all DCI integrands have two-particle cuts; the first examples without two-particle
cuts are the second and third graphs in fig.~\ref{four_loop_integrand} (which therefore are the first graphs ``allowed''
by two-particle cut considerations
to have coefficients other than $+1$).
We have verified the consistency of our results with this principle; at seven-loops we find that 814 of the
2329 coefficients correspond to integrals which admit a two-particle cut.

\appendix

\section{Appendix}

In this appendix we collect a few useful graph
theoretic results about
DCI diagrams.

\begin{lemma}
\label{lemma:one}
Every four-point DCI diagram has an overall factor of $x_{13} x_{24}$.
\end{lemma}

\begin{proof}
We will show that the overall factor of $x_{13} x_{24}$ is required in order
to pass the finiteness criterion of section~\ref{sec:divergent}.
Consider an $L$-loop pseudo-conformal diagram with $P$ propagators (internal edges).
By dimensional analysis there must be a total of
$P-2L$ numerator factors.
Now consider the behavior of the integrand
in the limit when all of the $x$'s (both internal and external) approach a common point $y$.
If $(x - y)^2 = \mathcal{O}(\rho^2)$ for all $x$'s then the integrand scales like
\begin{equation}
\label{eq:logdiv}
\sim \frac{\rho^{2(P-2L)}}{\rho^{2P}} \rho^{4L-1} d\rho = \frac{d \rho}{\rho}
\end{equation}
where we have included the appropriate factor $\rho^{4L-1} d\rho$ from the integration measure
in radial coordinates near $\rho=0$.

Of course the four $x$'s associated with the external faces are supposed
to be fixed, so the fact that we have a divergence in this hypothetical limit is not yet of concern.
If we remove a single external face $x_a$ from consideration (that is, hold $x_a$ fixed instead of letting it approach $y$),
then the propagators adjacent to $x_a$ no longer contribute a factor of $\rho^{-2}$ each to eqn.~(\ref{eq:logdiv}), but
the numerator factors attached to $x_a$ also no longer contribute a factor of $\rho^{+2}$.  Since there are an equal number
of each, there is no change in the conclusion and we still have a log-divergent integral.

If we remove two or more external faces from consideration this analysis is unchanged unless at some step we encounter
a numerator factor $x_{ab}^m$ connecting two external edges.  In order to avoid counting the no-longer-contributing
factor of $\rho^{2m}$ twice we should multiply eqn.~(\ref{eq:logdiv}) by $\rho^{2m}$, which renders the integral finite for any $m>0$.

Therefore, by the time we've excluded three external faces from the limit (thereby reaching an honest limit of the integral),
we must have encountered at least
one numerator factor between two external faces in order to avoid a divergent integral.  Since the diagram must converge for
every possible limit in order to be DCI, we conclude that for any choice of three external faces $x_a,x_b,x_c$, two of them must
always share a numerator factor (that is, there must be a positive power of at least one of $x_{ab}$, $x_{bc}$ or $x_{ca}$).
Since $x_{ab} = 0$ if $a$ and $b$ are adjacent, the only way this is possible is for the diagram to contain
an overall factor of $x_{13} x_{24}$.
\end{proof}

\begin{figure}\caption{A degenerate graph is one in which two external legs are attached to the same vertex.
Such graphs can be pseudo-conformal (although this example is not), but they cannot be DCI.}
\begin{center}
\includegraphics[scale=0.4]{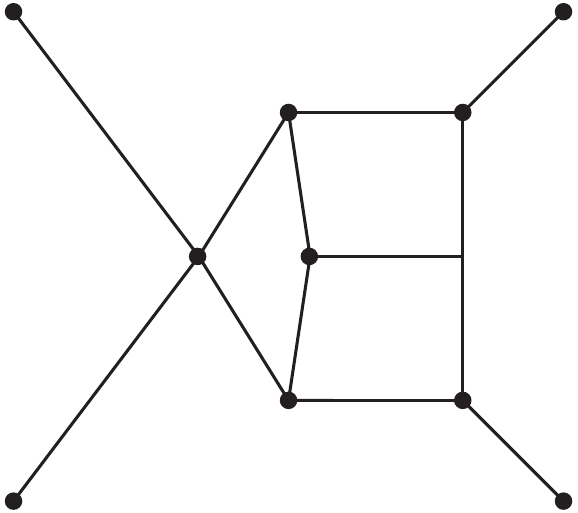}\vspace{-0.4cm}
\end{center}
\label{fig:degenerate}
\end{figure}

Since each of the four external faces $x_1,\ldots,x_4$ appears at least once in the numerator of every DCI integrand,
and conformal invariance requires that each one appears exactly as often in the denominator as in the numerator, we have
the immediate

\begin{corollary}
Degenerate diagrams (see figure~\ref{fig:degenerate}) cannot be DCI.
\end{corollary}

\section*{Note Added}

While we have checked our conjecture explicitly through seven loops, we are aware that starting at eight loops there are DCI integrals which do not diverge as $\mathcal{O}(1/\epsilon^2)$  in the limit~(\ref{eq:mainlimit}) and are therefore not detectable using our method. It would be interesting to see if there are further criteria that could determine the coefficients of these potential contributions to the four-particle amplitude at higher loop orders.

\begin{figure}\begin{center}\caption{The integrands which generate the four-loop amplitude.\label{four_loop_integrand}}~\\\hspace{-0.85cm}\begin{tabular}{m{4cm}@{}m{4cm}@{}m{4cm}@{}m{4cm}@{}m{4cm}}
\includegraphics[width=3.75cm]{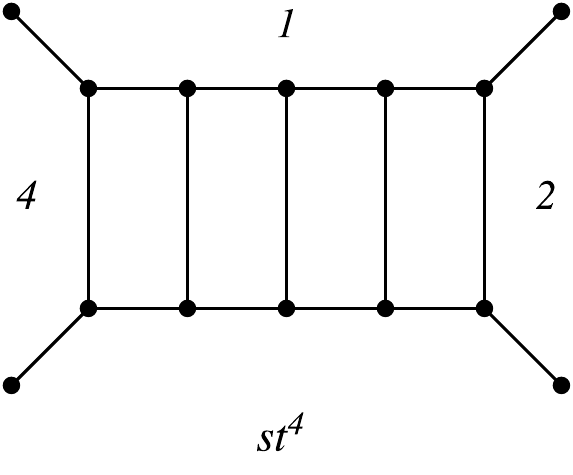}&
\includegraphics[width=3.75cm]{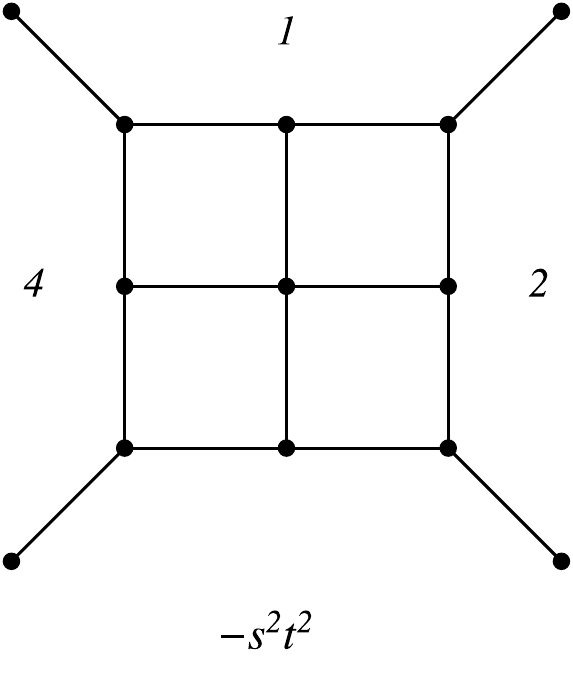}&
\includegraphics[width=3.75cm]{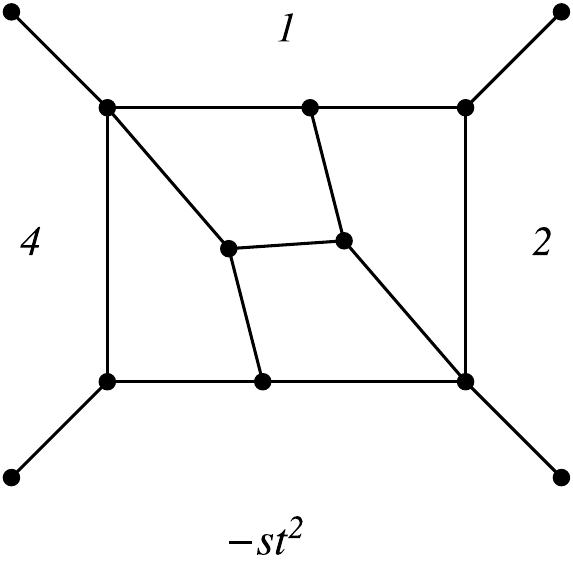}&
\includegraphics[width=3.75cm]{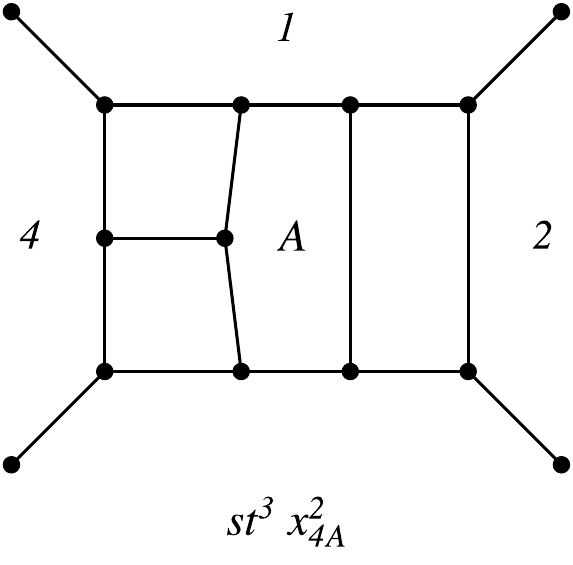}\\
\includegraphics[width=3.75cm]{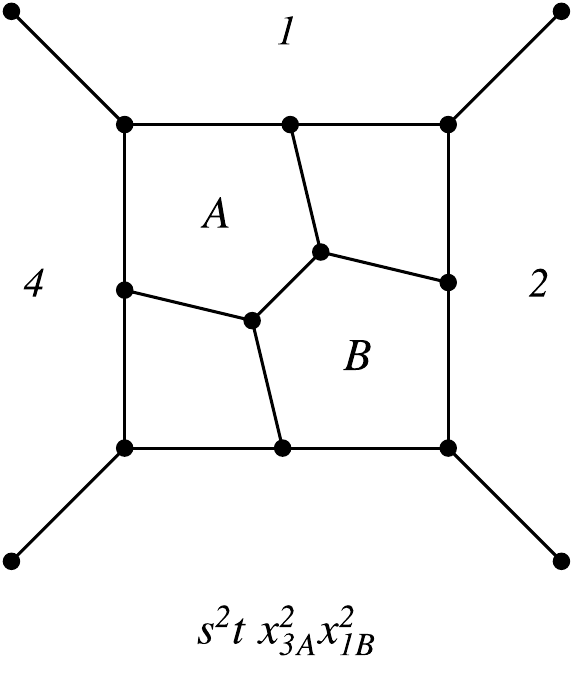}&
\includegraphics[width=3.75cm]{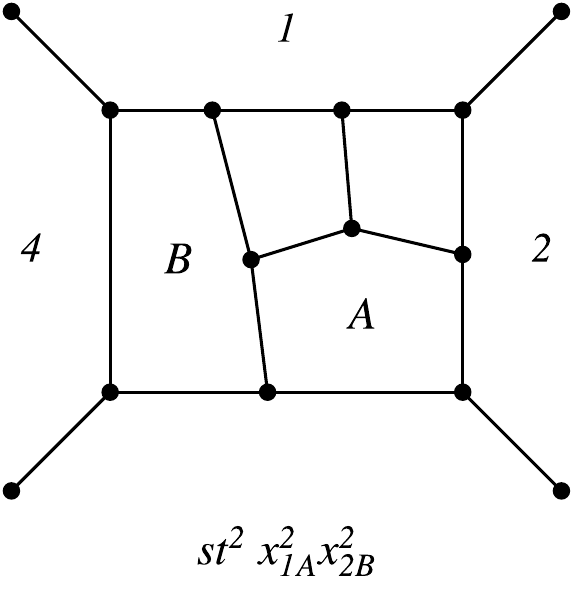}&
\includegraphics[width=3.75cm]{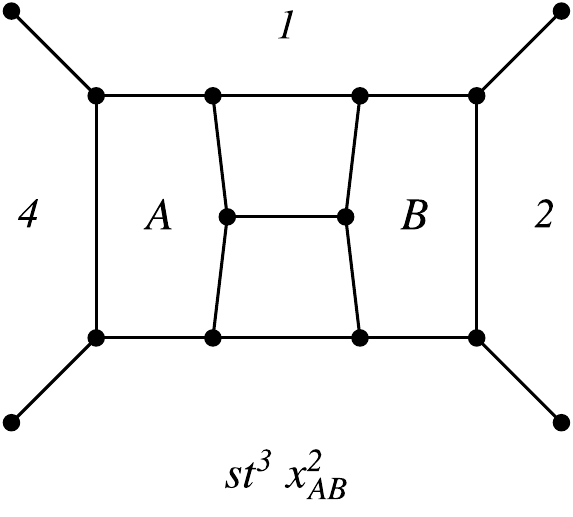}&
\includegraphics[width=3.75cm]{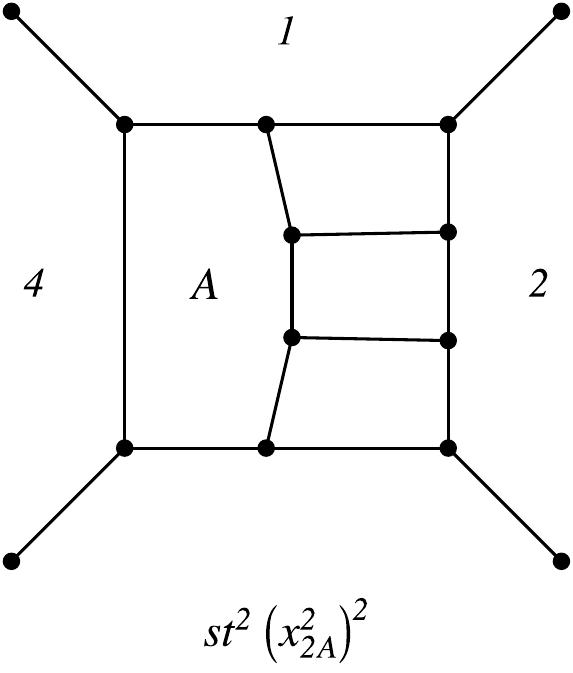}
\end{tabular}
\end{center}\vspace{-1cm}
\end{figure}

\begin{figure}[]\vspace{-1.05cm}\begin{center}\hspace{-0.05cm}\begin{tabular}{m{3cm}@{}m{3cm}@{}m{3cm}@{}m{3cm}@{}m{3cm}}
\includegraphics[width=2.9cm]{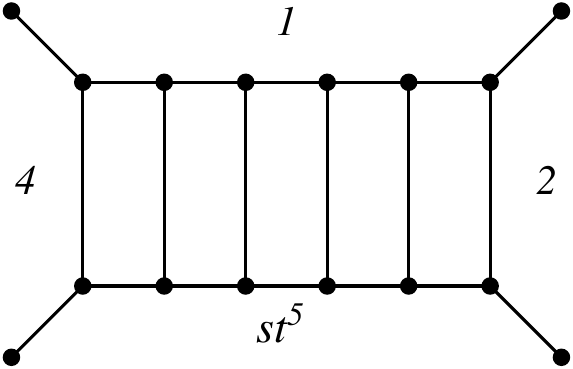}&
\includegraphics[width=2.9cm]{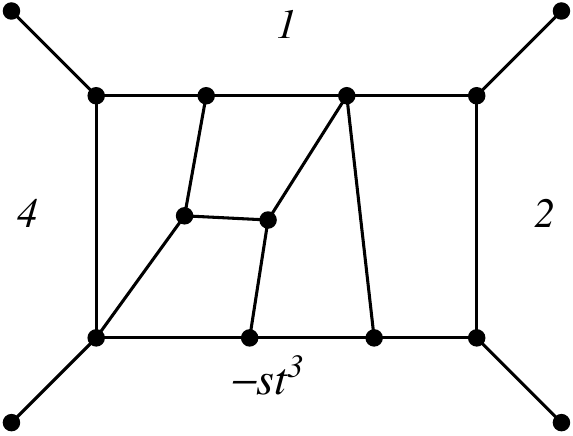}&
\includegraphics[width=2.9cm]{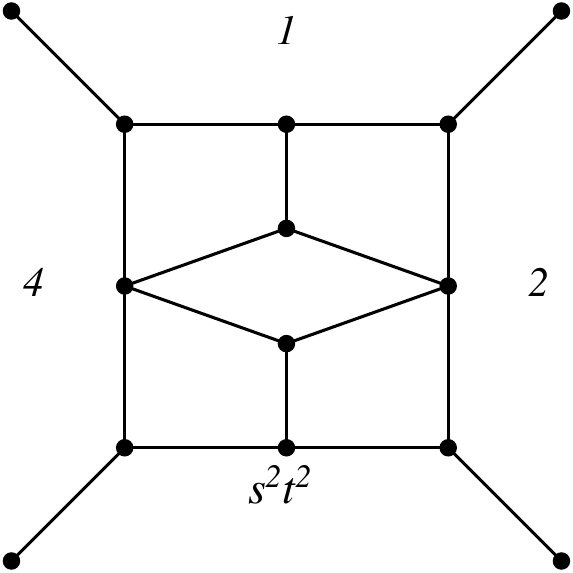}&
\includegraphics[width=2.9cm]{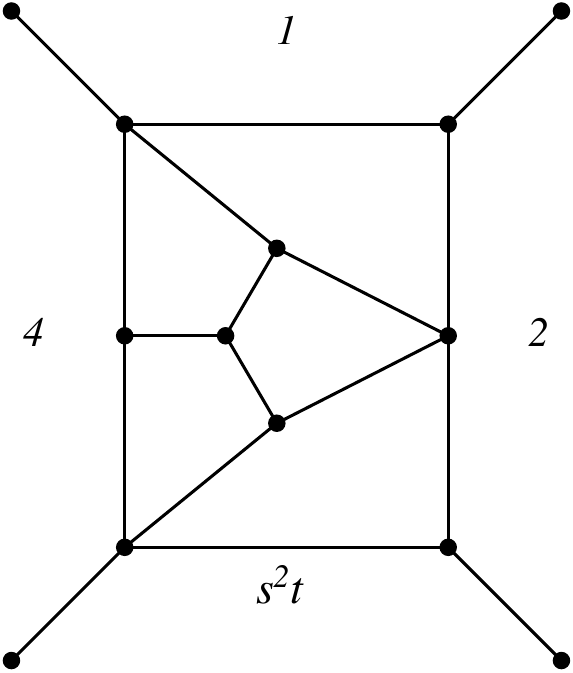}&
\includegraphics[width=2.9cm]{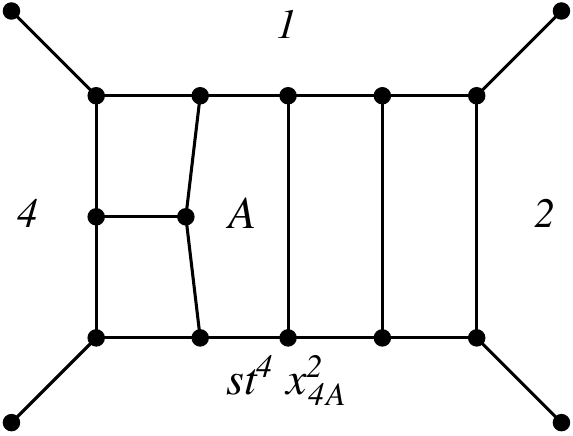}\\~\\[-0.4cm]
\includegraphics[width=2.9cm]{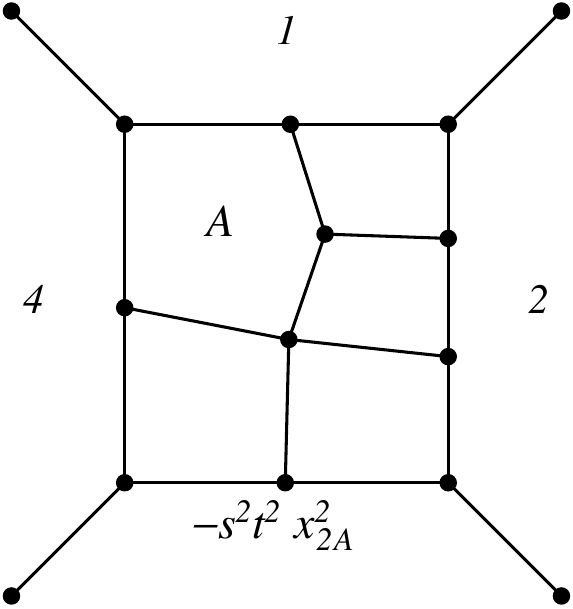}&
\includegraphics[width=2.9cm]{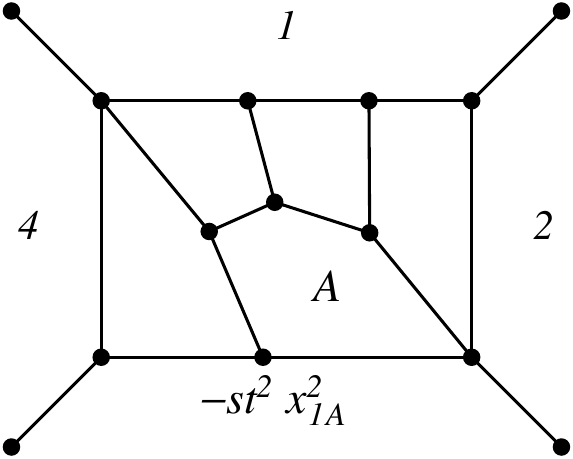}&
\includegraphics[width=2.9cm]{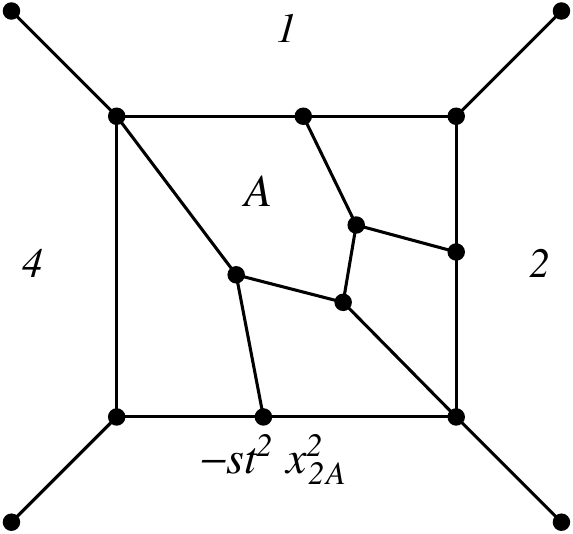}&
\includegraphics[width=2.9cm]{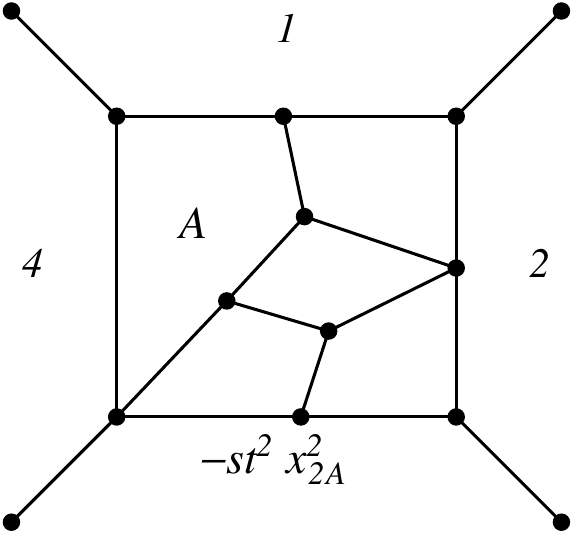}&
\includegraphics[width=2.9cm]{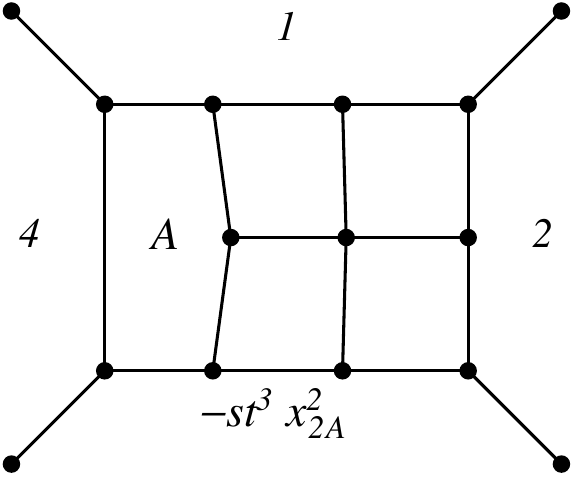}\\~\\[-0.4cm]
\includegraphics[width=2.9cm]{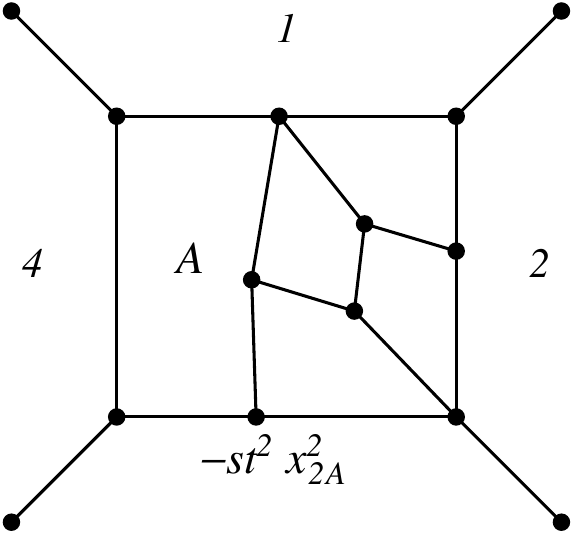}&
\includegraphics[width=2.9cm]{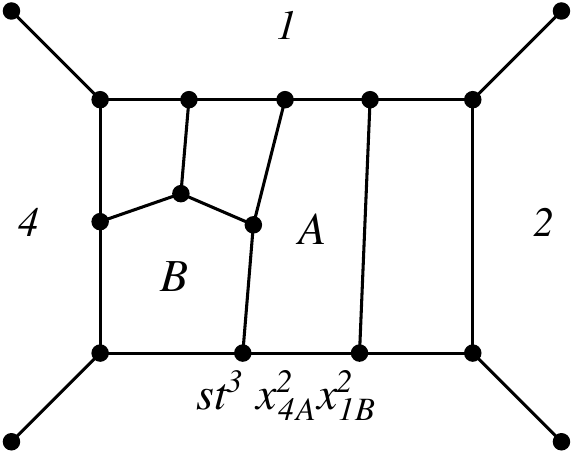}&
\includegraphics[width=2.9cm]{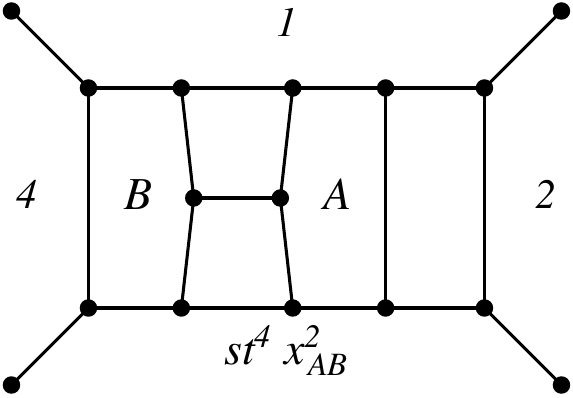}&
\includegraphics[width=2.9cm]{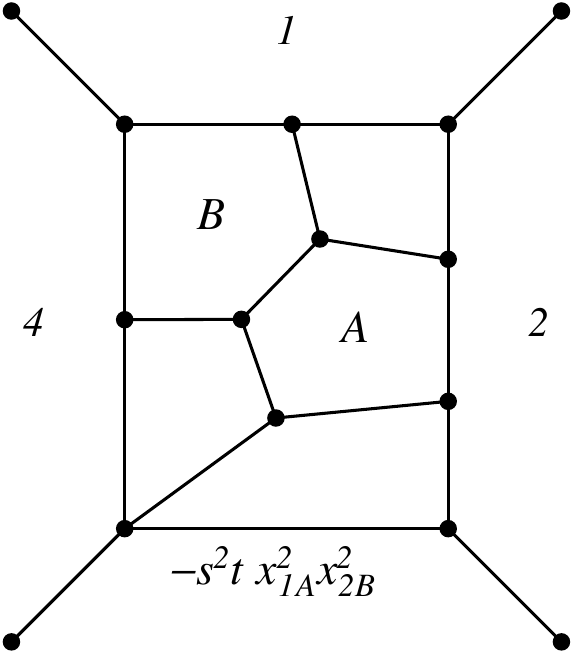}&
\includegraphics[width=2.9cm]{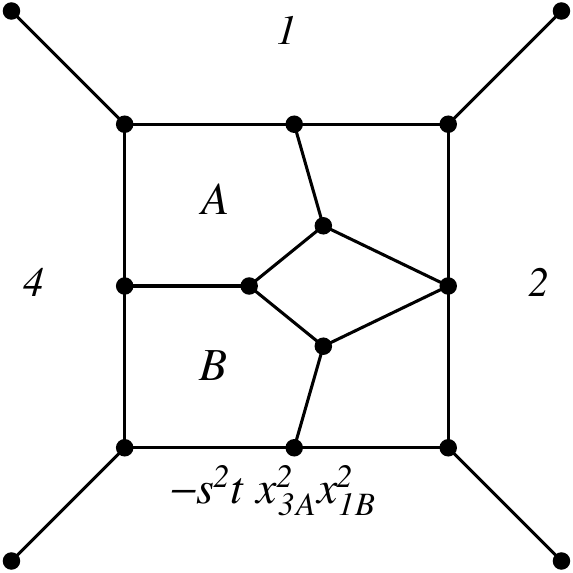}\\~\\[-0.4cm]
\includegraphics[width=2.9cm]{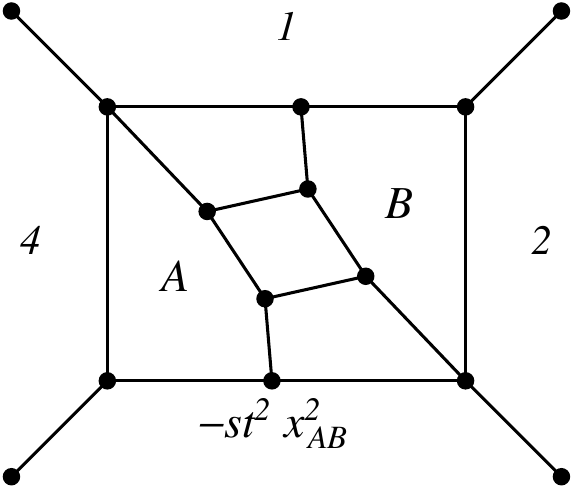}&
\includegraphics[width=2.9cm]{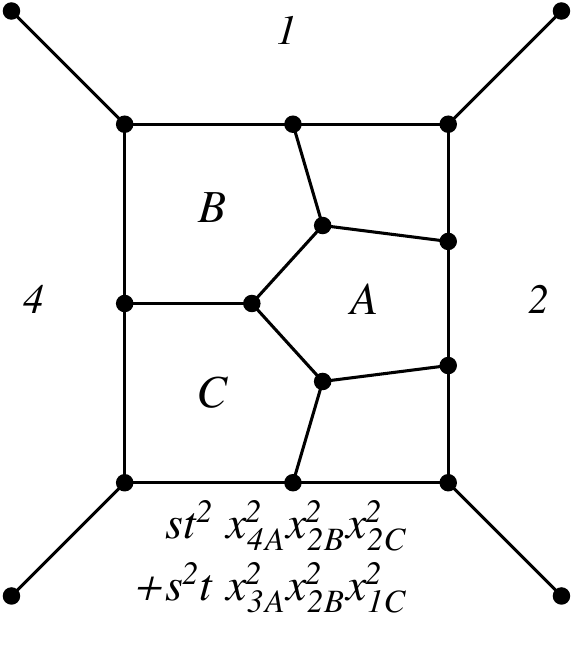}&
\includegraphics[width=2.9cm]{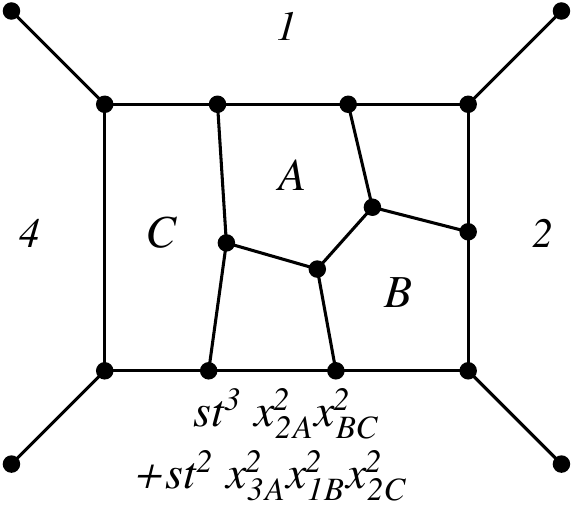}&
\includegraphics[width=2.9cm]{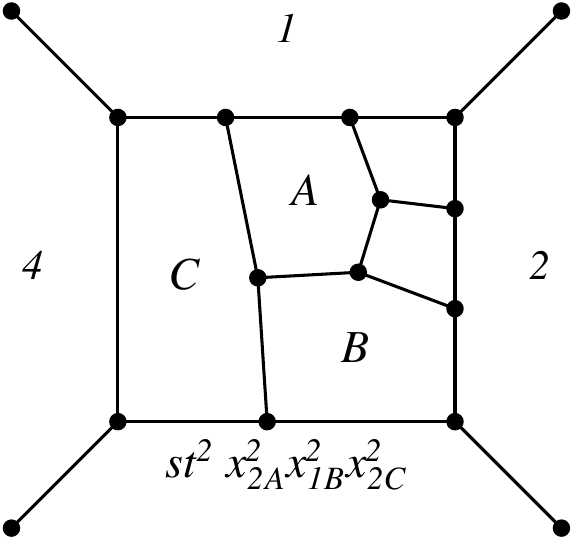}&
\includegraphics[width=2.9cm]{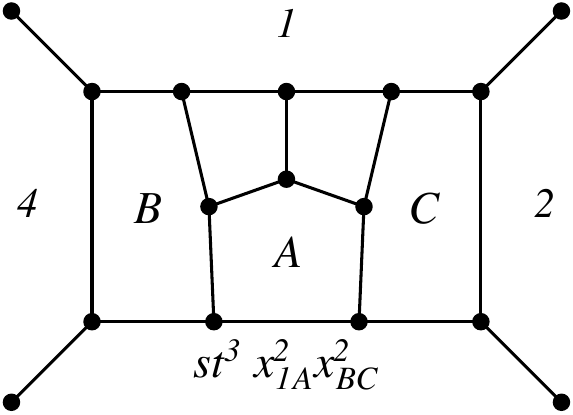}\\~\\[-0.4cm]
\includegraphics[width=2.9cm]{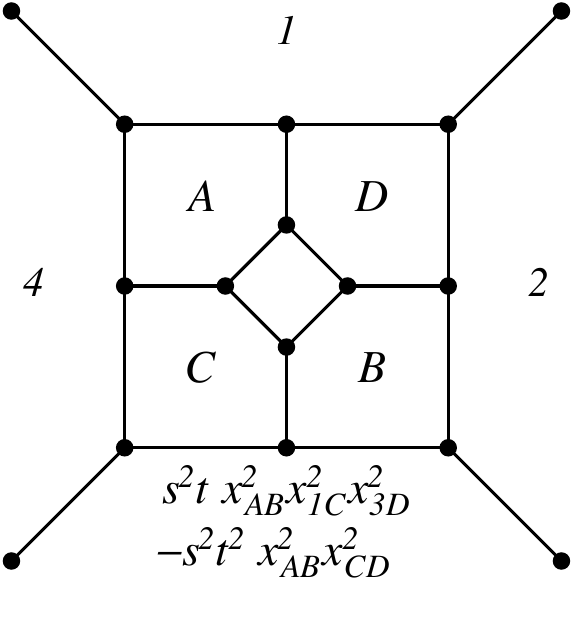}&
\includegraphics[width=2.9cm]{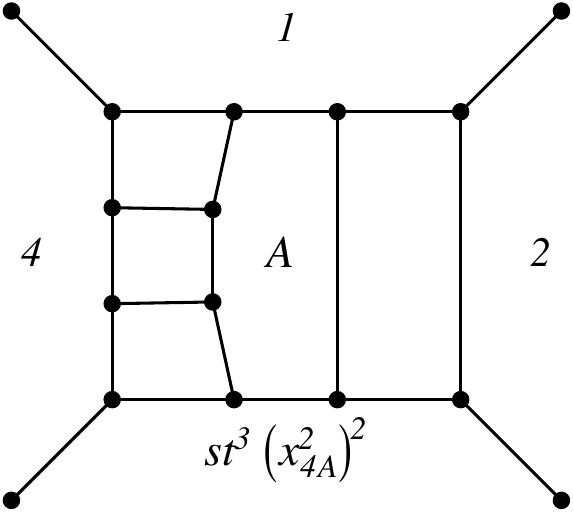}&
\includegraphics[width=2.9cm]{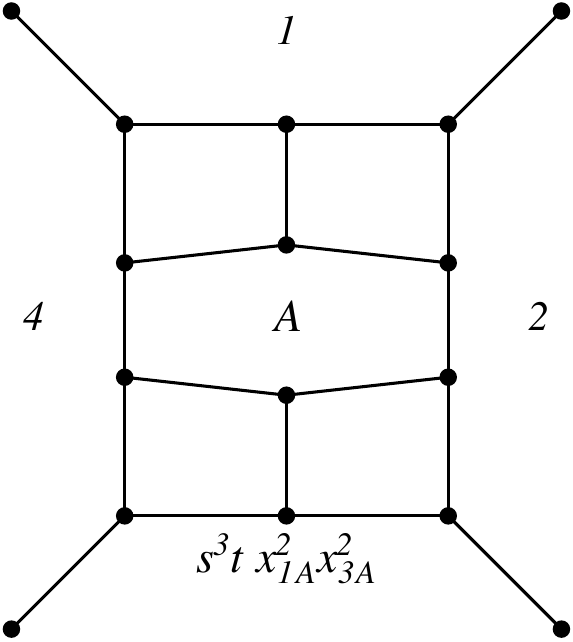}&
\includegraphics[width=2.9cm]{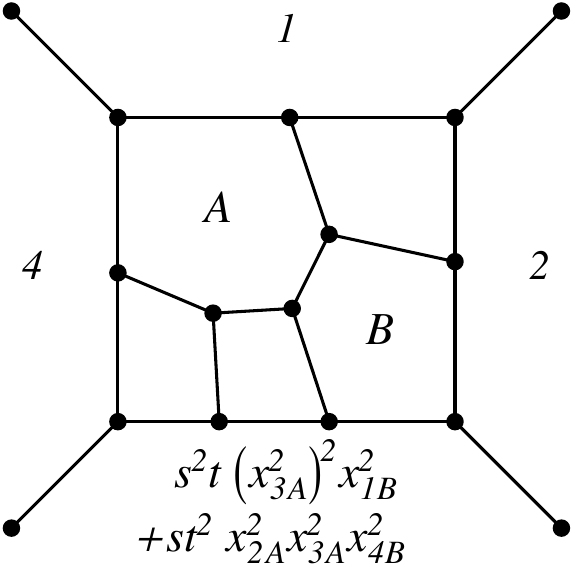}&
\includegraphics[width=2.9cm]{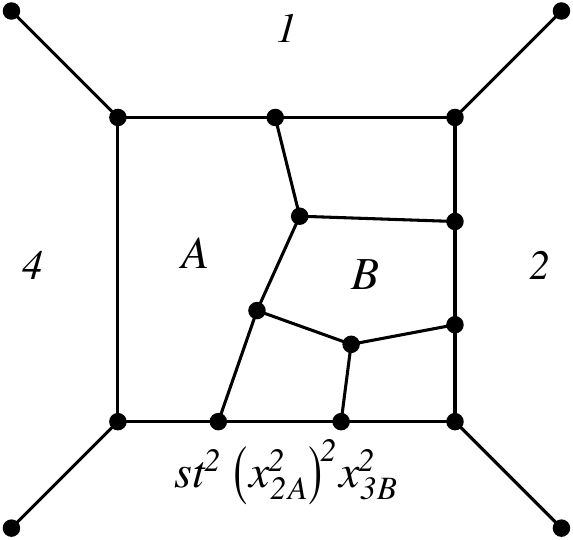}\\~\\[-0.4cm]
\includegraphics[width=2.9cm]{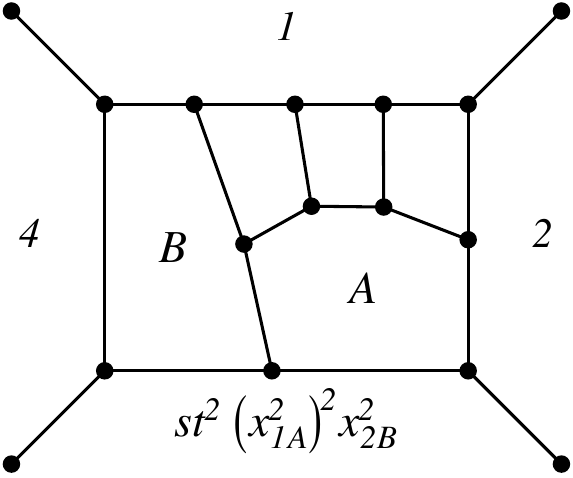}&
\includegraphics[width=2.9cm]{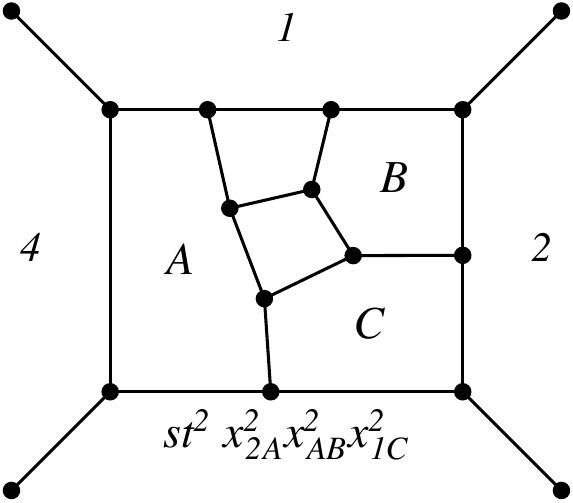}&
\includegraphics[width=2.9cm]{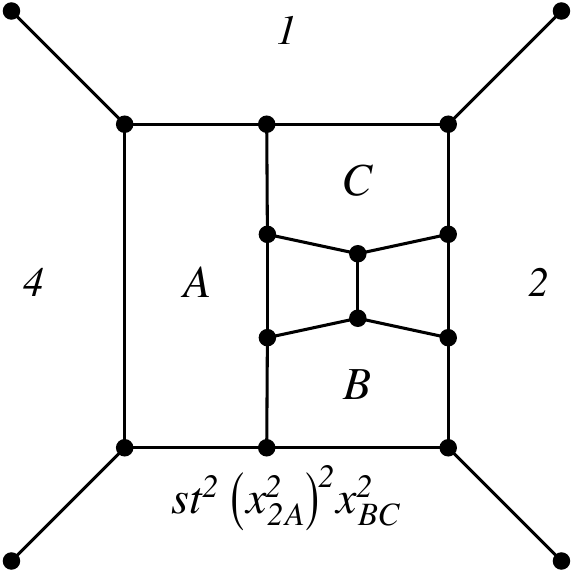}&
\includegraphics[width=2.9cm]{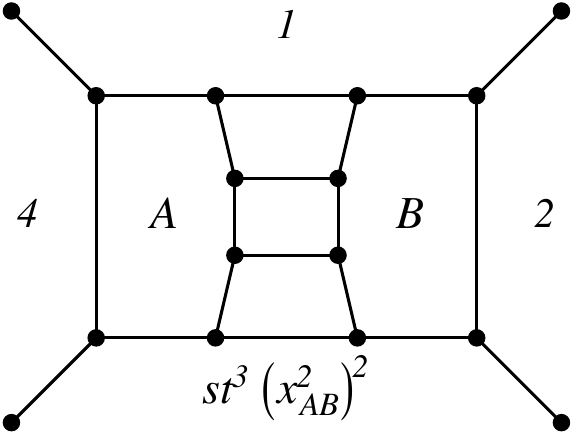}&
\includegraphics[width=2.9cm]{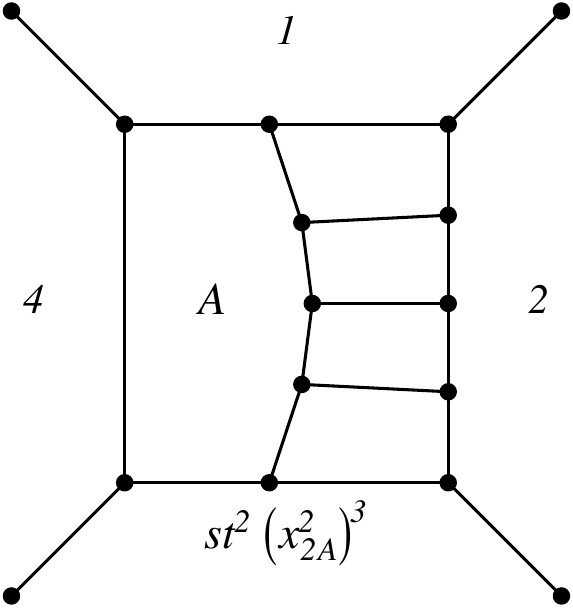}
\end{tabular}
\caption{The integrands which generate the five-loop amplitude.\label{five_loop_integrand}}\end{center}\vspace{-0cm}
\end{figure}

\begin{figure}[t]\vspace{-1cm}\scalebox{1}{\begin{minipage}[H!]{1\textwidth}
\hspace{-0.00\textwidth}
\includegraphics[width=1\textwidth]{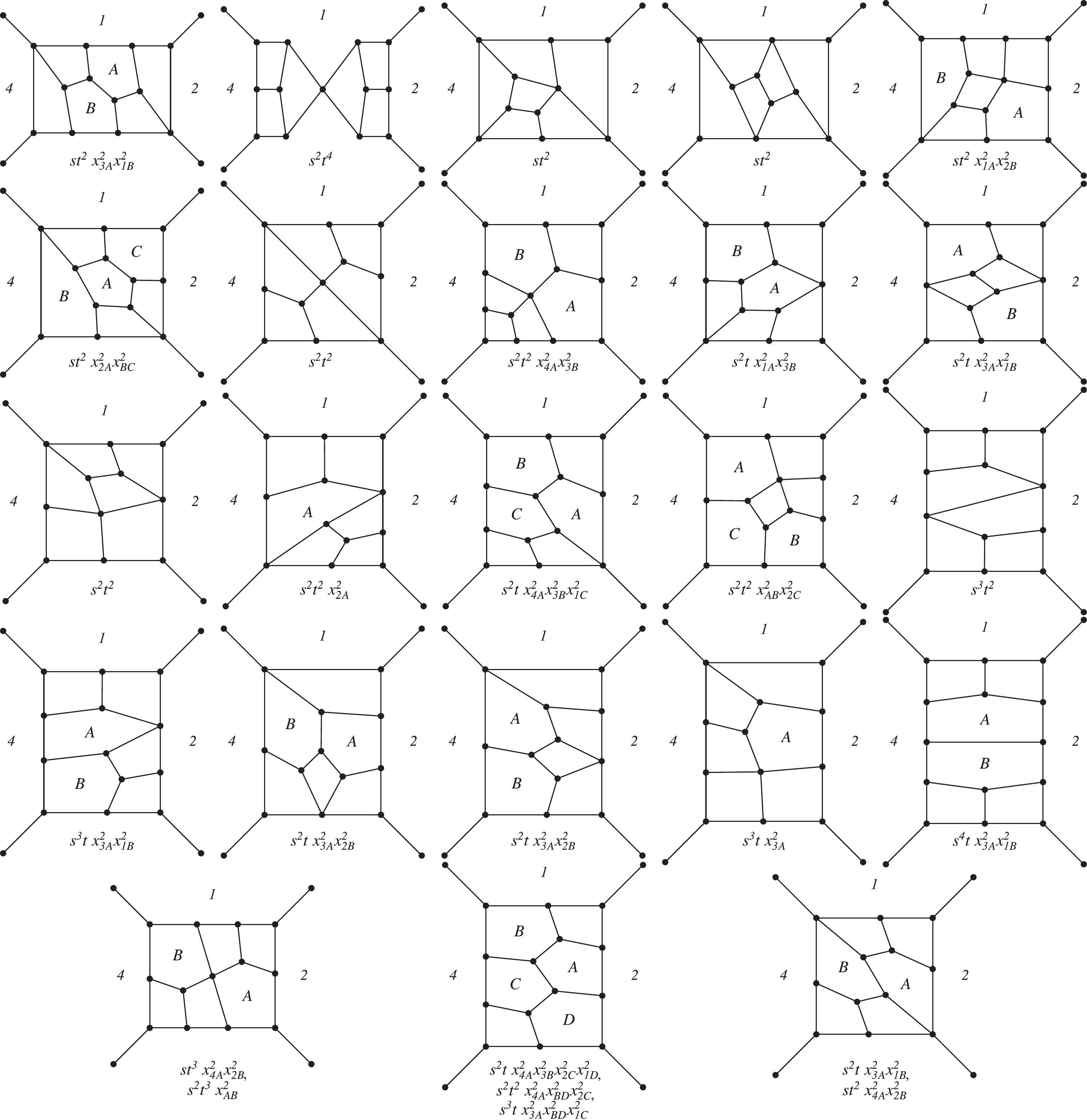}
\caption{The 27 six-loop integrands with vanishing coefficients. Notice that each of the last three graphs admit more than one numerator; the coefficient of each possible integrand is separately zero.\label{six_loop_coeff_0}}\end{minipage}\vspace{-5.5cm}}
\end{figure}

\section*{Acknowledgments}
It is a pleasure to acknowledge the many helpful insights offered by F.~Cachazo, J.~Trnka, and especially N.~Arkani-Hamed during the earliest stages of this work. We are also grateful for discussions with and encouragement from  G.~Korchemsky,
D.~Skinner and E.~Sokatchev. We are also grateful for Enrico Hermann for pointing out some important typographical errors in earlier versions of this paper. This work was supported in part by the Harvard Society of Fellows (JB); the US Department of Energy, under contracts DE-FG02-91ER40654 (JB), DE-FG02-91ER40688 (MS and AV) and DE-FG02-11ER41742 Early Career Award (AV); the National Science Foundation, under grants PHY-0756966 (JB) and PHY-0643150 PECASE (AV); and the Sloan Research Foundation (AV).

\providecommand{\href}[2]{#2}\begingroup\raggedright\endgroup


\begin{thebibliography}{10}

\bibitem{Bern:1994cg}
Z.~Bern, L.~J. Dixon, D.~C. Dunbar, and D.~A. Kosower, ``{Fusing Gauge Theory
  Tree Amplitudes into Loop Amplitudes},''
  \href{http://dx.doi.org/10.1016/0550-3213(94)00488-Z}{{\em Nucl. Phys.} {\bf
  B435} (1995)  59--101},
\href{http://arxiv.org/abs/hep-ph/9409265}{{ arXiv:hep-ph/9409265}}.

\bibitem{Bern:1994zx}
Z.~Bern, L.~J. Dixon, D.~C. Dunbar, and D.~A. Kosower, ``{One-Loop $n$-Point
  Gauge Theory Amplitudes, Unitarity and Collinear Limits},''
  \href{http://dx.doi.org/10.1016/0550-3213(94)90179-1}{{\em Nucl. Phys.} {\bf
  B425} (1994)  217--260},
\href{http://arxiv.org/abs/hep-ph/9403226}{{ arXiv:hep-ph/9403226}}.

\bibitem{Bern:2007ct}
Z.~Bern, J.~Carrasco, H.~Johansson, and D.~Kosower, ``{Maximally Supersymmetric
  Planar Yang-Mills Amplitudes at Five Loops},''
  \href{http://dx.doi.org/10.1103/PhysRevD.76.125020}{{\em Phys. Rev.} {\bf
  D76} (2007)  125020},
\href{http://arxiv.org/abs/0705.1864}{{ arXiv:0705.1864 [hep-th]}}.

\bibitem{Buchbinder:2005wp}
E.~I. Buchbinder and F.~Cachazo, ``{Two-Loop Amplitudes of Gluons and Octa-Cuts
  in $\mathcal{N}\!=\!4$ super Yang-Mills},''
  \href{http://dx.doi.org/10.1088/1126-6708/2005/11/036}{{\em JHEP} {\bf 0511}
  (2005)  036},
\href{http://arxiv.org/abs/hep-th/0506126}{{ arXiv:hep-th/0506126}}.

\bibitem{Cachazo:2008vp}
F.~Cachazo, ``{Sharpening The Leading Singularity},''
\href{http://arxiv.org/abs/0803.1988}{{ arXiv:0803.1988 [hep-th]}}.

\bibitem{Cachazo:2008hp}
F.~Cachazo, M.~Spradlin, and A.~Volovich, ``{Leading Singularities of the
  Two-Loop Six-Particle MHV Amplitude},''
  \href{http://dx.doi.org/10.1103/PhysRevD.78.105022}{{\em Phys. Rev.} {\bf
  D78} (2008)  105022},
\href{http://arxiv.org/abs/0805.4832}{{ arXiv:0805.4832 [hep-th]}}.

\bibitem{Spradlin:2008uu}
M.~Spradlin, A.~Volovich, and C.~Wen, ``{Three-Loop Leading Singularities and
  BDS Ansatz for Five Particles},''
  \href{http://dx.doi.org/10.1103/PhysRevD.78.085025}{{\em Phys. Rev.} {\bf
  D78} (2008)  085025},
\href{http://arxiv.org/abs/0808.1054}{{ arXiv:0808.1054 [hep-th]}}.

\bibitem{Spradlin:1900zz}
M.~Spradlin, ``{Multiloop Gluon Amplitudes and AdS/CFT},''
{\em In the Proceedings of 9th Workshop on Non-Perturbative Quantum
  Chromodynamics, Paris, France} {\bf C0706044} (2007)  05.

\bibitem{Bern:2011qt}
Z.~Bern and Y.-t. Huang, ``{Basics of Generalized Unitarity},''
  \href{http://dx.doi.org/10.1088/1751-8113/44/45/454003}{{\em J. Phys.} {\bf
  A44} (2011)  454003},
\href{http://arxiv.org/abs/1103.1869}{{ arXiv:1103.1869 [hep-th]}}.

\bibitem{unpublished_six_loop}
Z.~Bern, L.~J. Dixon, J.~J.~M. Carrasco, and H.~Johansson unpublished.

\bibitem{CaronHuot:2010zt}
S.~Caron-Huot, ``{Loops and Trees},''
  \href{http://dx.doi.org/10.1007/JHEP05(2011)080}{{\em JHEP} {\bf 1105} (2011)
   080},
\href{http://arxiv.org/abs/1007.3224}{{ arXiv:1007.3224 [hep-ph]}}.

\bibitem{Boels:2010nw}
R.~H. Boels, ``{On BCFW Shifts of Integrands and Integrals},''
  \href{http://dx.doi.org/10.1007/JHEP11(2010)113}{{\em JHEP} {\bf 1011} (2010)
   113},
\href{http://arxiv.org/abs/1008.3101}{{ arXiv:1008.3101 [hep-th]}}.

\bibitem{ArkaniHamed:2010kv}
N.~Arkani-Hamed, J.~L. Bourjaily, F.~Cachazo, S.~Caron-Huot, and J.~Trnka,
  ``{The All-Loop Integrand For Scattering Amplitudes in Planar
  $\mathcal{N}\!=\!4$ SYM},''
  \href{http://dx.doi.org/10.1007/JHEP01(2011)041}{{\em JHEP} {\bf 1101} (2011)
   041},
\href{http://arxiv.org/abs/1008.2958}{{ arXiv:1008.2958 [hep-th]}}.

\bibitem{Eden:2010zz}
B.~Eden, G.~P. Korchemsky, and E.~Sokatchev, ``{From Correlation Functions to
  Scattering Amplitudes},''
  \href{http://dx.doi.org/10.1007/JHEP12(2011)002}{{\em JHEP} {\bf 1112} (2011)
   002},
\href{http://arxiv.org/abs/1007.3246}{{ arXiv:1007.3246 [hep-th]}}.

\bibitem{Alday:2010zy}
L.~F. Alday, B.~Eden, G.~P. Korchemsky, J.~Maldacena, and E.~Sokatchev, ``{From
  Correlation Functions to Wilson Loops},''
  \href{http://dx.doi.org/10.1007/JHEP09(2011)123}{{\em JHEP} {\bf 1109} (2011)
   123},
\href{http://arxiv.org/abs/1007.3243}{{ arXiv:1007.3243 [hep-th]}}.

\bibitem{Eden:2010ce}
B.~Eden, G.~P. Korchemsky, and E.~Sokatchev, ``{More on the Duality
  Correlators/Amplitudes},''
  \href{http://dx.doi.org/10.1016/j.physletb.2012.02.014}{{\em Phys. Lett.}
  {\bf B709} (2012)  247--253},
\href{http://arxiv.org/abs/1009.2488}{{ arXiv:1009.2488 [hep-th]}}.

\bibitem{Eden:2011yp}
B.~Eden, P.~Heslop, G.~P. Korchemsky, and E.~Sokatchev, ``{The
  Super-Correlator/Super-Amplitude Duality: Part I},''
  \href{http://dx.doi.org/10.1016/j.nuclphysb.2012.12.015}{{\em Nucl. Phys.}
  {\bf B869} (2013)  329--377},
\href{http://arxiv.org/abs/1103.3714}{{ arXiv:1103.3714 [hep-th]}}.

\bibitem{Eden:2011ku}
B.~Eden, P.~Heslop, G.~P. Korchemsky, and E.~Sokatchev, ``{The
  Super-Correlator/Super-Amplitude Duality: Part II},''
  \href{http://dx.doi.org/10.1016/j.nuclphysb.2012.12.014}{{\em Nucl. Phys.}
  {\bf B869} (2013)  378--416},
\href{http://arxiv.org/abs/1103.4353}{{ arXiv:1103.4353 [hep-th]}}.

\bibitem{Eden:2011we}
B.~Eden, P.~Heslop, G.~P. Korchemsky, and E.~Sokatchev, ``{Hidden Symmetry of
  Four-Point Correlation Functions and Amplitudes in $\mathcal{N}\!=\!4$
  SYM},'' \href{http://dx.doi.org/10.1016/j.nuclphysb.2012.04.007}{{\em Nucl.
  Phys.} {\bf B862} (2012)  193--231},
\href{http://arxiv.org/abs/1108.3557}{{ arXiv:1108.3557 [hep-th]}}.

\bibitem{Eden:2012tu}
B.~Eden, P.~Heslop, G.~P. Korchemsky, and E.~Sokatchev, ``{Constructing the
  Correlation Function of Four Stress-Tensor Multiplets and the Four-Particle
  Amplitude in $\mathcal{N}\!=\!4$ SYM},''
  \href{http://dx.doi.org/10.1016/j.nuclphysb.2012.04.013}{{\em Nucl. Phys.}
  {\bf B862} (2012)  450--503},
\href{http://arxiv.org/abs/1201.5329}{{ arXiv:1201.5329 [hep-th]}}.

\bibitem{Akhoury:1978vq}
R.~Akhoury, ``{Mass Divergences of Wide Angle Scattering Amplitudes},''
\href{http://dx.doi.org/10.1103/PhysRevD.19.1250}{{\em Phys. Rev.} {\bf D19}
  (1979)  1250}.

\bibitem{Mueller:1979ih}
A.~H. Mueller, ``{On the Asymptotic Behavior of the Sudakov Form-factor},''
\href{http://dx.doi.org/10.1103/PhysRevD.20.2037}{{\em Phys. Rev.} {\bf D20}
  (1979)  2037}.

\bibitem{Collins:1980ih}
J.~C. Collins, ``{Algorithm to Compute Corrections to the Sudakov
  Form-Factor},''
\href{http://dx.doi.org/10.1103/PhysRevD.22.1478}{{\em Phys. Rev.} {\bf D22}
  (1980)  1478}.

\bibitem{Sen:1982bt}
A.~Sen, ``{Asymptotic Behavior of the Wide Angle On-Shell Quark Scattering
  Amplitudes in Nonabelian Gauge Theories},''
\href{http://dx.doi.org/10.1103/PhysRevD.28.860}{{\em Phys. Rev.} {\bf D28}
  (1983)  860}.

\bibitem{Sterman:1986aj}
G.~F. Sterman, ``{Summation of Large Corrections to Short Distance Hadronic
  Cross-Sections},''
\href{http://dx.doi.org/10.1016/0550-3213(87)90258-6}{{\em Nucl. Phys.} {\bf
  B281} (1987)  310}.

\bibitem{Catani:1989ne}
S.~Catani and L.~Trentadue, ``{Resummation of the QCD Perturbative Series for
  Hard Processes},''
\href{http://dx.doi.org/10.1016/0550-3213(89)90273-3}{{\em Nucl. Phys.} {\bf
  B327} (1989)  323}.

\bibitem{Collins:1989bt}
J.~C. Collins, ``{Sudakov Form-Factors},'' {\em Adv. Ser. Direct. High Energy
  Phys.} {\bf 5} (1989)  573--614,
\href{http://arxiv.org/abs/hep-ph/0312336}{{ arXiv:hep-ph/0312336 [hep-ph]}}.

\bibitem{Magnea:1990zb}
L.~Magnea and G.~F. Sterman, ``{Analytic Continuation of the Sudakov
  Form-Factor in QCD},''
\href{http://dx.doi.org/10.1103/PhysRevD.42.4222}{{\em Phys. Rev.} {\bf D42}
  (1990)  4222--4227}.

\bibitem{Catani:1990rp}
S.~Catani and L.~Trentadue, ``{Comment on QCD Exponentiation at Large $x$},''
\href{http://dx.doi.org/10.1016/0550-3213(91)90506-S}{{\em Nucl. Phys.} {\bf
  B353} (1991)  183--186}.

\bibitem{Giele:1991vf}
W.~Giele and E.~N. Glover, ``{Higher Order Corrections to Jet Cross-Sections in
  $e^+\,e^-$ Annihilation},''
\href{http://dx.doi.org/10.1103/PhysRevD.46.1980}{{\em Phys. Rev.} {\bf D46}
  (1992)  1980--2010}.

\bibitem{Kunszt:1994np}
Z.~Kunszt, A.~Signer, and Z.~Trocsanyi, ``{Singular Terms of Helicity
  Amplitudes at One-Loop in QCD and the Soft Limit of the Cross-Sections of
  Multiparton Processes},''
  \href{http://dx.doi.org/10.1016/0550-3213(94)90077-9}{{\em Nucl. Phys.} {\bf
  B420} (1994)  550--564},
\href{http://arxiv.org/abs/hep-ph/9401294}{{ arXiv:hep-ph/9401294}}.

\bibitem{Catani:1998bh}
S.~Catani, ``{The Singular Behavior of QCD Amplitudes at Two Loop Order},''
  \href{http://dx.doi.org/10.1016/S0370-2693(98)00332-3}{{\em Phys. Lett.} {\bf
  B427} (1998)  161--171},
\href{http://arxiv.org/abs/hep-ph/9802439}{{ arXiv:hep-ph/9802439 [hep-ph]}}.

\bibitem{Vogt:2000ci}
A.~Vogt, ``{Next-to-Next-to-Leading Logarithmic Threshold Resummation for Deep
  Inelastic Scattering and the Drell-Yan Process},''
  \href{http://dx.doi.org/10.1016/S0370-2693(00)01344-7}{{\em Phys. Lett.} {\bf
  B497} (2001)  228--234},
\href{http://arxiv.org/abs/hep-ph/0010146}{{ arXiv:hep-ph/0010146 [hep-ph]}}.

\bibitem{Sterman:2002qn}
G.~F. Sterman and M.~E. Tejeda-Yeomans, ``{Multiloop Amplitudes and
  Resummation},'' \href{http://dx.doi.org/10.1016/S0370-2693(02)03100-3}{{\em
  Phys. Lett.} {\bf B552} (2003)  48--56},
\href{http://arxiv.org/abs/hep-ph/0210130}{{ arXiv:hep-ph/0210130 [hep-ph]}}.

\bibitem{Alday:2009zm}
L.~F. Alday, J.~M. Henn, J.~Plefka, and T.~Schuster, ``{Scattering into the
  Fifth Dimension of $\mathcal{N}\!=\!4$ super Yang-Mills},''
  \href{http://dx.doi.org/10.1007/JHEP01(2010)077}{{\em JHEP} {\bf 1001} (2010)
   077},
\href{http://arxiv.org/abs/0908.0684}{{ arXiv:0908.0684 [hep-th]}}.

\bibitem{Henn:2010bk}
J.~M. Henn, S.~G. Naculich, H.~J. Schnitzer, and M.~Spradlin,
  ``{Higgs-Regularized Three-Loop Four-Gluon Amplitude in $\mathcal{N}\!=\!4$
  SYM: Exponentiation and Regge Limits},''
  \href{http://dx.doi.org/10.1007/JHEP04(2010)038}{{\em JHEP} {\bf 04} (2010)
  038},
\href{http://arxiv.org/abs/1001.1358}{{ arXiv:1001.1358 [hep-th]}}.

\bibitem{Henn:2010ir}
J.~M. Henn, S.~G. Naculich, H.~J. Schnitzer, and M.~Spradlin, ``{More Loops and
  Legs in Higgs-Regulated $\mathcal{N}\!=\!4$ SYM Amplitudes},''
  \href{http://dx.doi.org/10.1007/JHEP08(2010)002}{{\em JHEP} {\bf 08} (2010)
  002},
\href{http://arxiv.org/abs/1004.5381}{{ arXiv:1004.5381 [hep-th]}}.

\bibitem{Bern:2006vw}
Z.~Bern, M.~Czakon, D.~Kosower, R.~Roiban, and V.~Smirnov, ``{Two-Loop
  Iteration of Five-Point $\mathcal{N}\!=\!4$ Super-Yang-Mills Amplitudes},''
  \href{http://dx.doi.org/10.1103/PhysRevLett.97.181601}{{\em Phys. Rev. Lett.}
  {\bf 97} (2006)  181601},
\href{http://arxiv.org/abs/hep-th/0604074}{{ arXiv:hep-th/0604074 [hep-th]}}.

\bibitem{Cachazo:2006tj}
F.~Cachazo, M.~Spradlin, and A.~Volovich, ``{Iterative Structure within the
  Five-Particle Two-Loop Amplitude},''
  \href{http://dx.doi.org/10.1103/PhysRevD.74.045020}{{\em Phys. Rev.} {\bf
  D74} (2006)  045020},
\href{http://arxiv.org/abs/hep-th/0602228}{{ arXiv:hep-th/0602228 [hep-th]}}.

\bibitem{Bern:2006ew}
Z.~Bern, M.~Czakon, L.~J. Dixon, D.~A. Kosower, and V.~A. Smirnov, ``{The
  Four-Loop Planar Amplitude and Cusp Anomalous Dimension in Maximally
  Supersymmetric Yang-Mills Theory},''
  \href{http://dx.doi.org/10.1103/PhysRevD.75.085010}{{\em Phys. Rev.} {\bf
  D75} (2007)  085010},
\href{http://arxiv.org/abs/hep-th/0610248}{{ arXiv:hep-th/0610248 [hep-th]}}.

\bibitem{Bern:2008ap}
Z.~Bern {\em et al.}, ``{The Two-Loop Six-Gluon MHV Amplitude in Maximally
  Supersymmetric Yang-Mills Theory},''
  \href{http://dx.doi.org/10.1103/PhysRevD.78.045007}{{\em Phys. Rev.} {\bf
  D78} (2008)  045007},
\href{http://arxiv.org/abs/0803.1465}{{ arXiv:0803.1465 [hep-th]}}.

\bibitem{Bern:2005iz}
Z.~Bern, L.~J. Dixon, and V.~A. Smirnov, ``{Iteration of Planar Amplitudes in
  Maximally Supersymmetric Yang-Mills Theory at Three Loops and Beyond},''
  \href{http://dx.doi.org/10.1103/PhysRevD.72.085001}{{\em Phys. Rev.} {\bf
  D72} (2005)  085001},
\href{http://arxiv.org/abs/hep-th/0505205}{{ arXiv:hep-th/0505205}}.

\bibitem{Drummond:2010mb}
J.~M. Drummond and J.~M. Henn, ``{Simple Loop Integrals and Amplitudes in
  $\mathcal{N}\!=\!4$ SYM},''
  \href{http://dx.doi.org/10.1007/JHEP05(2011)105}{{\em JHEP} {\bf 1105} (2011)
   105},
\href{http://arxiv.org/abs/1008.2965}{{ arXiv:1008.2965 [hep-th]}}.

\bibitem{ArkaniHamed:2010gh}
N.~Arkani-Hamed, J.~L. Bourjaily, F.~Cachazo, and J.~Trnka, ``{Local Integrals
  for Planar Scattering Amplitudes},''
  \href{http://dx.doi.org/10.1007/JHEP06(2012)125}{{\em JHEP} {\bf 1206} (2012)
   125},
\href{http://arxiv.org/abs/1012.6032}{{ arXiv:1012.6032 [hep-th]}}.

\bibitem{Bargheer:2009qu}
T.~Bargheer, N.~Beisert, W.~Galleas, F.~Loebbert, and T.~McLoughlin,
  ``{Exacting $\mathcal{N}\!=\!4$ Superconformal Symmetry},''
  \href{http://dx.doi.org/10.1088/1126-6708/2009/11/056}{{\em JHEP} {\bf 11}
  (2009)  056},
\href{http://arxiv.org/abs/0905.3738}{{ arXiv:0905.3738 [hep-th]}}.

\bibitem{Drummond:2007cf}
J.~M. Drummond, J.~Henn, G.~P. Korchemsky, and E.~Sokatchev, ``{On Planar Gluon
  Amplitudes/Wilson Loops Duality},''
  \href{http://dx.doi.org/10.1016/j.nuclphysb.2007.11.007}{{\em Nucl. Phys.}
  {\bf B795} (2008)  52--68},
\href{http://arxiv.org/abs/0709.2368}{{ arXiv:0709.2368 [hep-th]}}.

\bibitem{Drummond:2007au}
J.~M. Drummond, J.~Henn, G.~P. Korchemsky, and E.~Sokatchev, ``{Conformal Ward
  Identities for Wilson Loops and a Test of the Duality with Gluon
  Amplitudes},'' \href{http://dx.doi.org/10.1016/j.nuclphysb.2009.10.013}{{\em
  Nucl. Phys.} {\bf B826} (2010)  337--364},
\href{http://arxiv.org/abs/0712.1223}{{ arXiv:0712.1223 [hep-th]}}.

\bibitem{Korchemsky:1985xj}
G.~Korchemsky and A.~Radyushkin, ``{Loop Space Formalism and Renormalization
  Group for the Infrared Asymptotics of {QCD}},''
\href{http://dx.doi.org/10.1016/0370-2693(86)91439-5}{{\em Phys. Lett.} {\bf
  B171} (1986)  459--467}.

\bibitem{Beisert:2006ez}
N.~Beisert, B.~Eden, and M.~Staudacher, ``{Transcendentality and Crossing},''
  \href{http://dx.doi.org/10.1088/1742-5468/2007/01/P01021}{{\em J. Stat.
  Mech.} {\bf 0701} (2007)  P01021},
\href{http://arxiv.org/abs/hep-th/0610251}{{ arXiv:hep-th/0610251 [hep-th]}}.

\bibitem{Mason:2010yk}
L.~Mason and D.~Skinner, ``{The Complete Planar $S$-Matrix of
  $\mathcal{N}\!=\!4$ SYM as a Wilson Loop in Twistor Space},''
  \href{http://dx.doi.org/10.1007/JHEP12(2010)018}{{\em JHEP} {\bf 12} (2010)
  018},
\href{http://arxiv.org/abs/1009.2225}{{ arXiv:1009.2225 [hep-th]}}.

\bibitem{CaronHuot:2010ek}
S.~Caron-Huot, ``{Notes on the Scattering Amplitude / Wilson Loop Duality},''
  \href{http://dx.doi.org/10.1007/JHEP07(2011)058}{{\em JHEP} {\bf 1107} (2011)
   058},
\href{http://arxiv.org/abs/1010.1167}{{ arXiv:1010.1167 [hep-th]}}.

\bibitem{Belitsky:2011zm}
A.~Belitsky, G.~Korchemsky, and E.~Sokatchev, ``{Are Scattering Amplitudes Dual
  to Super Wilson Loops?},''
  \href{http://dx.doi.org/10.1016/j.nuclphysb.2011.10.014}{{\em Nucl. Phys.}
  {\bf B855} (2012)  333--360},
\href{http://arxiv.org/abs/1103.3008}{{ arXiv:1103.3008 [hep-th]}}.

\bibitem{Adamo:2011pv}
T.~Adamo, M.~Bullimore, L.~Mason, and D.~Skinner, ``{Scattering Amplitudes and
  Wilson Loops in Twistor Space},''
  \href{http://dx.doi.org/10.1088/1751-8113/44/45/454008}{{\em J. Phys.} {\bf
  A44} (2011)  454008},
\href{http://arxiv.org/abs/1104.2890}{{ arXiv:1104.2890 [hep-th]}}.

\bibitem{four_point_multiloop_datafiles}
 PDF and {\sc Mathematica} files containing all our results are provided at
  \href{http://goo.gl/qIKe8}{{\tt http://goo.gl/qIKe8}.}

\bibitem{Hodges:2009hk}
A.~Hodges, ``{Eliminating Spurious Poles from Gauge-Theoretic Amplitudes},''
\href{http://arxiv.org/abs/0905.1473}{{ arXiv:0905.1473 [hep-th]}}.

\bibitem{Bern:2010tq}
Z.~Bern, J.~Carrasco, L.~J. Dixon, H.~Johansson, and R.~Roiban, ``{The Complete
  Four-Loop Four-Point Amplitude in $\mathcal{N}\!=\!4$ Super-Yang-Mills
  Theory},'' \href{http://dx.doi.org/10.1103/PhysRevD.82.125040}{{\em Phys.
  Rev.} {\bf D82} (2010)  125040},
\href{http://arxiv.org/abs/1008.3327}{{ arXiv:1008.3327 [hep-th]}}.

\bibitem{Bern:2012uf}
Z.~Bern, J.~Carrasco, L.~Dixon, H.~Johansson, and R.~Roiban, ``{Simplifying
  Multiloop Integrands and Ultraviolet Divergences of Gauge Theory and Gravity
  Amplitudes},'' \href{http://dx.doi.org/10.1103/PhysRevD.85.105014}{{\em Phys.
  Rev.} {\bf D85} (2012)  105014},
\href{http://arxiv.org/abs/1201.5366}{{ arXiv:1201.5366 [hep-th]}}.

\bibitem{Green:1982sw}
M.~B. Green, J.~H. Schwarz, and L.~Brink, ``{$\mathcal{N}\!=\!4$ Yang-Mills and
  $\mathcal{N}\!=\!8$ Supergravity as Limits of String Theories},''
\href{http://dx.doi.org/10.1016/0550-3213(82)90336-4}{{\em Nucl. Phys.} {\bf
  B198} (1982)  474--492}.

\bibitem{Anastasiou:2003kj}
C.~Anastasiou, Z.~Bern, L.~J. Dixon, and D.~A. Kosower, ``{Planar Amplitudes in
  Maximally Supersymmetric Yang-Mills Theory},''
  \href{http://dx.doi.org/10.1103/PhysRevLett.91.251602}{{\em Phys. Rev. Lett.}
  {\bf 91} (2003)  251602},
\href{http://arxiv.org/abs/hep-th/0309040}{{ arXiv:hep-th/0309040}}.

\bibitem{Gluza:2010ws}
J.~Gluza, K.~Kajda, and D.~A. Kosower, ``{Towards a Basis for Planar Two-Loop
  Integrals},'' \href{http://dx.doi.org/10.1103/PhysRevD.83.045012}{{\em Phys.
  Rev.} {\bf D83} (2011)  045012},
\href{http://arxiv.org/abs/1009.0472}{{ arXiv:1009.0472 [hep-th]}}.

\bibitem{Drummond:2006rz}
J.~Drummond, J.~Henn, V.~Smirnov, and E.~Sokatchev, ``{Magic Identities for
  Conformal Four-Point Integrals},''
  \href{http://dx.doi.org/10.1088/1126-6708/2007/01/064}{{\em JHEP} {\bf 0701}
  (2007)  064},
\href{http://arxiv.org/abs/hep-th/0607160}{{ arXiv:hep-th/0607160}}.

\bibitem{Drummond:2007aua}
J.~M. Drummond, G.~P. Korchemsky, and E.~Sokatchev, ``{Conformal Properties of
  Four-Gluon Planar Amplitudes and Wilson loops},''
  \href{http://dx.doi.org/10.1016/j.nuclphysb.2007.11.041}{{\em Nucl. Phys.}
  {\bf B795} (2008)  385--408},
\href{http://arxiv.org/abs/0707.0243}{{ arXiv:0707.0243 [hep-th]}}.

\bibitem{Nguyen:2007ya}
D.~Nguyen, M.~Spradlin, and A.~Volovich, ``{New Dual Conformally Invariant
  Off-Shell Integrals},''
  \href{http://dx.doi.org/10.1103/PhysRevD.77.025018}{{\em Phys. Rev.} {\bf
  D77} (2008)  025018},
\href{http://arxiv.org/abs/0709.4665}{{ arXiv:0709.4665 [hep-th]}}.

\bibitem{Nogueira:1991ex}
P.~Nogueira, ``{Automatic Feynman Graph Generation},''
\href{http://dx.doi.org/10.1006/jcph.1993.1074}{{\em J. Comput. Phys.} {\bf
  105} (1993)  279--289}.

\bibitem{qgraf}
 QGRAF, \href{http://cfif.ist.utl.pt/~paulo/qgraf.html}{{\tt http://cfif.ist.utl.pt/$\sim$paulo/qgraf.html}}.

\bibitem{plantripaper}
G.~Brinkmann and B.~D. McKay, ``{Fast Generation of Planar Graphs (expanded
  Version)}.'' \href{http://cs.anu.edu.au/~bdm/papers/plantri-full.pdf}{{\tt http://cs.anu.edu.au/$\sim$bdm/papers/plantri-full.pdf}}.

\bibitem{plantri}
 plantri and fullgen, \href{http://cs.anu.edu.au/~bdm/plantri/}{{\tt http://cs.anu.edu.au/$\sim$bdm/plantri/}}.

\bibitem{Cachazo:2008dx}
F.~Cachazo and D.~Skinner, ``{On the Structure of Scattering Amplitudes in
  $\mathcal{N}\!=\!4$ Super Yang-Mills and $\mathcal{N}\!=\!8$ Supergravity},''
\href{http://arxiv.org/abs/0801.4574}{{ arXiv:0801.4574 [hep-th]}}.

\bibitem{Bern:1997nh}
Z.~Bern, J.~Rozowsky, and B.~Yan, ``{Two-Loop Four-Gluon Amplitudes in
  $\mathcal{N}\!=\!4$ SuperYang-Mills},''
  \href{http://dx.doi.org/10.1016/S0370-2693(97)00413-9}{{\em Phys. Lett.} {\bf
  B401} (1997)  273--282},
\href{http://arxiv.org/abs/hep-ph/9702424}{{ arXiv:hep-ph/9702424}}.

\end{thebibliography}
\end{document}